\documentclass[11pt]{article}

\usepackage[T1]{fontenc}
\usepackage{lmodern}
\usepackage{amsmath,amssymb,amsthm,mathtools,bm}
\usepackage{fullpage}
\usepackage{booktabs}
\usepackage{longtable}
\usepackage{array}
\usepackage{tikz}
\usetikzlibrary{arrows.meta,positioning}
\usepackage[numbers,sort&compress]{natbib}
\usepackage{microtype}
\usepackage{mathpazo}
\usepackage{authblk}
\usepackage{xcolor}
\usepackage[
  colorlinks=true,
  linkcolor=blue,
  citecolor=blue,
  urlcolor=blue
]{hyperref}
\hypersetup{
  pdftitle={A Sharp Local-Question Threshold for GHZ-Equatorial Completeness in Four-Player XOR Games},
  pdfauthor={Anonymous},
  pdfsubject={The local-question threshold for GHZ equatorial strategies in perfect four-player XOR games},
  pdfkeywords={XOR games, commuting-operator strategies, GHZ strategies, Cayley games, primitive circuits, Magnus expansion}
}

\newtheorem{theorem}{Theorem}[section]
\newtheorem{proposition}[theorem]{Proposition}
\newtheorem{lemma}[theorem]{Lemma}
\newtheorem{corollary}[theorem]{Corollary}
\theoremstyle{definition}
\newtheorem{definition}[theorem]{Definition}
\theoremstyle{remark}

\newcommand{\F}{\mathbb{F}}
\newcommand{\Z}{\mathbb{Z}}
\newcommand{\Q}{\mathbb{Q}}
\newcommand{\R}{\mathbb{R}}
\newcommand{\one}{\mathbf{1}}
\newcommand{\supp}{\operatorname{supp}}
\newcommand{\rank}{\operatorname{rank}}
\newcommand{\dist}{\operatorname{d}}

\newcommand{\co}{\mathrm{co}}
\newcommand{\unif}{\mathrm{unif}}
\newcommand{\card}[1]{\lvert #1\rvert}
\newcommand{\inner}[2]{\left\langle #1,#2\right\rangle}
\allowdisplaybreaks
\numberwithin{equation}{section}
\begin{document}

\title{A Sharp Local-Question Threshold for GHZ-Equatorial Completeness in Four-Player XOR Games}

\author[1]{Ziao Tang}

\author[2]{Chengkai Zhu}

\author[1]{Ge Bai}

\author[1,*]{Xin Wang}

\author[3,$\dagger$]{Ranyiliu Chen}

\affil[1]{The Hong Kong University of Science and Technology (Guangzhou),
Guangdong 511453, China}

\affil[2]{QudeLeap Research, Shanghai 200030, China}

\affil[3]{Quantum Science Center of Guangdong-Hong Kong-Macao Greater Bay Area,
Shenzhen 518045, China}
\date{}

\maketitle
\begingroup
\renewcommand{\thefootnote}{\fnsymbol{footnote}}

\footnotetext[1]{\href{mailto:felixxinwang@hkust-gz.edu.cn}
{felixxinwang@hkust-gz.edu.cn}}

\footnotetext[2]{\href{mailto:chenranyiliu@quantumsc.cn}
{chenranyiliu@quantumsc.cn}}

\endgroup
\begin{abstract}

We determine the smallest number of active questions per player at which a
four-player binary exclusive-or (XOR) game of commuting-operator value one
need not admit a Greenberger--Horne--Zeilinger (GHZ) equatorial realization.
Such a realization uses the four-qubit GHZ state and equatorial qubit
observables, reducing perfect play to additive phase equations. We prove that
every four-player XOR game with commuting-operator value one and at most three
active questions per player has a perfect GHZ-equatorial strategy. Conversely,
we construct a uniform eight-clause game with four active questions per player
whose commuting-operator value is one but whose phase equations are
inconsistent. Thus four is the sharp local-question threshold. The positive
result follows by lifting every integral incidence obstruction to an ordered
noncommutative refutation, using primitive circuits, forest matchings, and
ternary Hamming geometry. For the separating game, a Klein four-group
incidence relation obstructs the phase system, while an even-subgroup normal
form and degree-one and degree-two Magnus coefficients exclude refutations of
arbitrary length.
\end{abstract}
\vspace{5mm}
\medskip
\noindent\textbf{Formal verification.}
The principal threshold theorem and its obstruction-space corollary have
been formalized in Lean~4.30.0~\cite{deMouraUllrich2021} using
Mathlib~4.30.0~\cite{mathlib2020} and the Lean-QIT
library~\cite{ZhuEtAl2026LeanQIT}. The resulting proof terms are checked by
the Lean kernel. The source code is available in the accompanying GitHub
repository~\cite{FourPlayerXORGames2026}.
Appendix~\ref{app:lean} gives the statement-by-statement correspondence
between the manuscript and the formal development.
\medskip
\newpage
\setcounter{tocdepth}{2}
\tableofcontents

\section{Introduction}\label{sec:introduction}

Nonlocal games recast Bell experiments as cooperative tests and provide a
common language for Bell inequalities, entangled strategies, and
operator-algebraic models of multipartite measurements
\cite{Bell1964,ClauserHorneShimonyHolt1969,
CleveHoyerTonerWatrous2004}.  In a binary \(n\)-player XOR game, the verifier
samples a question tuple (or a \emph{clause}), receives one bit from each player, and checks a
prescribed parity.  
Despite this
restriction, multiplayer XOR games have served as a tractable testbed for
multipartite Bell violations and quantitative bounds on entangled values
\cite{PerezGarciaEtAl2008,BrietVidick2013}. 

At value one, every positive-weight clause imposes an exact algebraic relation among
the players' binary observables.  We compare two realizations of these
relations.  A finite-dimensional tensor-product strategy assigns a Hilbert
space factor to each player, whereas a commuting-operator strategy represents
all observables on one, possibly infinite-dimensional, Hilbert space and
requires operators belonging to distinct players to commute.  The latter
model contains the former and is central to semidefinite and game-algebraic
approaches to nonlocal correlations
\cite{NavascuesPironioAcin2008,CleveLiuSlofstra2017}.  Since the two correlation
models differ in general \cite{Slofstra2019,JiNatarajanVidickWrightYuen2020}, it
is a substantive restriction to ask when commuting-operator satisfiability is
already witnessed by a specific finite-dimensional construction.

For XOR games, the relevant construction is suggested by the GHZ--Mermin
paradox and the structure of multipartite full-correlation Bell inequalities
\cite{GreenbergerHorneShimonyZeilinger1990,Mermin1990,WernerWolf2001}.  We use
GHZ-equatorial, or maximally entangled relative-phase (MERP), strategies in the
terminology of \cite{WattsHarrowKanwarNatarajan2019}.  With the shared state and local measurements given by
\[
  \lvert\mathrm{GHZ}_n\rangle
  =\frac{\lvert0\rangle^{\otimes n}+\lvert1\rangle^{\otimes n}}{\sqrt2},
  \qquad
  M(\theta)=\cos\theta\,X+\sin\theta\,Y,
\]
one has
\begin{equation}\label{eq:intro-ghz-correlator}
  \bigl\langle\mathrm{GHZ}_n\bigr\rvert
  \bigotimes_{\alpha=1}^{n}M(\theta_\alpha)
  \bigl\lvert\mathrm{GHZ}_n\bigr\rangle
  =\cos\!\left(\sum_{\alpha=1}^{n}\theta_\alpha\right).
\end{equation}
Thus a clause \(c\) with target parity \(b_c\) is satisfied with certainty if
and only if
\begin{equation}\label{eq:intro-merp-phase}
  \sum_{\alpha=1}^{n}\theta_{\alpha,q_c^{(\alpha)}}
  =\pi b_c\pmod{2\pi}.
\end{equation}
When the angles are restricted to multiples of \(\pi/2\),
this recovers the usual \(X/Y\) relations of the GHZ--Mermin construction.
Arbitrary angles extend them to continuous \(U(1)\) phase data.  The phase
description is abelian, although the local observables need not commute;
indeed, $M(\theta)M(\phi)=e^{i(\phi-\theta)Z}.$ MERP nevertheless fixes the local dimension, shared state, and measurement
geometry, while a general commuting-operator strategy imposes none of these
restrictions.
The question is whether Eq.~\eqref{eq:intro-merp-phase} is complete for perfect
commuting-operator satisfiability.

Prior work identifies two complementary frontiers of MERP completeness.
Werner and Wolf's analysis of multipartite full-correlation inequalities with
two dichotomic observables per site underlies the two-question regime: for any
number of players, commuting-operator value one is witnessed by a perfect
GHZ-equatorial strategy when each player has at most two active questions
\cite{WernerWolf2001,WattsHelton2023}.  Bene Watts et al. formulated the MERP--PREF duality, which turns consistency of the phase system into the absence of an integer incidence obstruction~\cite{WattsHarrowKanwarNatarajan2019}.  The same work showed that completeness
does not hold in general: it identified a six-player \(123\) game which has a perfect finite-dimensional strategy but no perfect MERP strategy.  In the other direction, Bene Watts and Helton proved that every three-player XOR game of commuting-operator value one admits a perfect three-qubit MERP strategy, for
arbitrary question alphabets, and obtained a polynomial-time decision
procedure for perfect \(3\)XOR~\cite{WattsHelton2023}.   Taken together, these results leave unresolved the
smallest active local alphabet at which perfect commuting-operator
satisfiability can escape the GHZ-phase model in the four-player setting.
This brings us to the central question of the paper:
\begin{quote}
\emph{What is the smallest value of \(k\) for which a four-player \(k\)-question XOR game
can have commuting-operator value one while admitting no perfect MERP
strategy?}
\end{quote}
The following theorem gives the exact answer.

\begin{theorem}\label{thm:main}

\emph{(i)} Let \(G\) be any finite four-player XOR game. If \(G\) has at most three active questions per player, then
\begin{equation}\label{eq:main-result}
  \omega_{\co}(G)=1
  \quad\Longleftrightarrow\quad
  G\text{ admits a perfect MERP strategy}.
\end{equation}

\emph{(ii)} There exists a four-player XOR game \(G\), with four active
questions per player, such that \(\omega_{\co}(G)=1\) but no perfect MERP
strategy exists.  Consequently, four is the smallest number of active
questions per player for which such a separation can occur.
\end{theorem}

We remark that part~(i) does not assert uniqueness or rigidity of the GHZ realization.
Whether this four-question game admits a perfect
finite-dimensional non-MERP strategy remains open as well.

The four-player problem is the natural next boundary case after the
three-player collapse theorem of Bene Watts and Helton
\cite{WattsHelton2023}, but Theorem~\ref{thm:main} is not obtained by a
direct extension of their proof.  That argument is intrinsically organized
around three page projections: projector and right-inverse constructions
clear two page components, while specialized gadget maps control the
remaining one.  With a fourth page, the same operations leave coupled
residual components on several pages, and the clearing identities no longer
close.  Moreover, part~(ii) shows that the unrestricted four-player analogue
is false. A different support-level lifting mechanism is therefore needed.

The algebraic content of this lifting mechanism is most transparent in terms
of two obstruction spaces.  An integer dependence
among the clause-incidence rows whose target-weighted parity is odd obstructs
the MERP phase equations; following
\cite{WattsHarrowKanwarNatarajan2019}, we call it a parity-permuted refutation
specification (PREF).  A true refutation is stronger: it orders clause
occurrences so that every player's question word freely reduces to the
identity and the total target parity is odd.  The known dualities say that a
PREF exists exactly when no perfect MERP strategy exists, whereas a true
refutation exists exactly when the commuting-operator value is below one
\cite{WattsHarrowKanwarNatarajan2019,WattsHelton2023,WattsHeltonKlep2023}.

The collapse result (part (i)) is then proved through an exact circuit-lifting theorem. We show that, for every four-partite ternary
support, each abelian obstruction can be lifted to an ordered
noncommutative refutation. In particular, the
integer relation lattice is generated by primitive support-minimal circuits,
and a rank bound confines each such circuit to at most ten clauses.  For every
primitive circuit we construct a balanced word with the prescribed absolute
multiplicities, using signed-occurrence forest matchings to produce the order
and ternary Hamming geometry to exclude the exceptional configurations.  The
ten-point extremal case is resolved by a ternary perfect-code argument. This
is a uniform lifting result rather than a bounded search. For sharpness (part (ii)), we prove that the Klein four-group game has a PREF but no true refutation of any length. An even-subgroup normal
form, together with the degree-one and degree-two Magnus coefficients, rules
out true refutations with arbitrary order and multiplicity by an explicit
integral parity certificate \cite{Magnus1937,MagnusKarrassSolitar1976}.

Section~\ref{sec:preliminaries} develops the two obstruction criteria and the
Klein witness.  Section~\ref{sec:sufficiency} proves exact circuit lifting for
four ternary pages.  Section~\ref{sec:sharpness} establishes sharpness through
the Cayley-game analysis and the all-length Magnus obstruction.
Section~\ref{sec:discussion} discusses scope and open problems. Additionally, a Lean~4 formalization of our main theorem is provided for completeness, which is given in Appendix \ref{app:lean}.

\section{Preliminaries and obstruction criteria}\label{sec:preliminaries}

\subsection{Games, strategies, and active supports}

A finite four-player binary XOR game \(G\) consists of finite question sets
\(Q_1,\ldots,Q_4\) and a finite family of clauses with weights \(p_c\geq0\)
satisfying \(\sum_c p_c=1\).  Clause \(c\) is specified
by a question tuple
\[
  q_c=(q_c^{(1)},\ldots,q_c^{(4)})
  \in Q_1\times\cdots\times Q_4,
\]
a probability \(p_c\), and a target parity \(b_c\in\mathbb F_2\).  On clause
\(c\), the players return bits \(a_1,\ldots,a_4\) and win precisely when
\[
  a_1\oplus\cdots\oplus a_4=b_c.
\]
A commuting-operator strategy consists of a complex Hilbert space
\(\mathcal H\), a unit vector \(\psi\in\mathcal H\), and bounded
self-adjoint involutions \(A_{\alpha,j}\in\mathcal B(\mathcal H)\) such that
\([A_{\alpha,j},A_{\beta,k}]=0\) whenever \(\alpha\ne\beta\).  Its success
probability is
\begin{equation}\label{eq:co-value-definition}
 V(G;\psi,A)=\sum_c\frac{p_c}{2}
 \left(1+(-1)^{b_c}
 \left\langle\psi,
 \prod_{\alpha=1}^{4}A_{\alpha,q_c^{(\alpha)}}\psi
 \right\rangle\right).
\end{equation}
The commuting-operator value \(\omega_{\co}(G)\) is the supremum of
Eq.~\eqref{eq:co-value-definition} over such strategies.  A strategy is
perfect precisely when
\(\prod_\alpha A_{\alpha,q_c^{(\alpha)}}\psi=(-1)^{b_c}\psi\) for every
positive-weight clause.

A question is \emph{active} if it occurs in at least one positive-weight
clause.  Thus the active question sets, rather than unused alphabet symbols,
are the local resources counted in Theorem~\ref{thm:main}.

Only the positive-weight support is relevant to perfection.  The resulting
reduction to a uniform distribution on the supported clauses is standard in
the perfect-value setting; see, for example,
\cite[Sec.~2.1.1]{WattsHelton2023}.  We delete zero-weight clauses and merge
repeated question tuples with the same target.  Repeated tuples with opposite
targets preclude a perfect strategy.  The resulting set
\[
  E\subseteq Q_1\times\cdots\times Q_4
\]
is the \emph{reduced support}.  Proposition~\ref{prop:weight-reduction} shows
that replacing its positive distribution by the uniform distribution
preserves both commuting-operator value one and perfect MERP satisfiability.
We may therefore work with distinct, uniformly weighted clauses whenever
only these exact properties are at issue.  We state the reduction with an
explicit comparison inequality and include the proof for completeness.

\begin{proposition}[Standard support and weight reduction]
\label{prop:weight-reduction}
Let a reduced support have \(m\) clauses, and let \(p_{\min}\) and
\(p_{\max}\) be the minimum and maximum of a strictly positive clause
distribution \(p\).  For every strategy,
\begin{equation}\label{eq:weight-comparison-app}
 m p_{\min}(1-V_{\unif})
 \leq 1-V_p
 \leq m p_{\max}(1-V_{\unif}),
\end{equation}
where \(V_p\) and \(V_{\unif}\) are its success probabilities under \(p\)
and the uniform distribution.  Hence \(\omega_{\co,p}=1\) if and only if
\(\omega_{\co,\unif}=1\).  Perfect MERP satisfiability likewise depends
only on the reduced support and its targets.
\end{proposition}

\begin{proof}
If \(\ell_c\in[0,1]\) is the loss probability on clause \(c\), then
\(1-V_p=\sum_c p_c\ell_c\) and
\(1-V_{\unif}=m^{-1}\sum_c\ell_c\).  Bounding every \(p_c\) between
\(p_{\min}\) and \(p_{\max}\) proves Eq.~\eqref{eq:weight-comparison-app};
taking infima over strategies proves the value-one equivalence.  A MERP
strategy is perfect precisely when it satisfies the phase equation on every
positive-weight clause.  Deleting zero-weight clauses and merging repeated
equal-target clauses therefore preserve both properties.  If two clauses on
the same question tuple have opposite targets and weights \(p_0,p_1>0\),
their combined loss is at least \(\min\{p_0,p_1\}\), independently of the
strategy.  Such a pair therefore forces \(\omega_{\co}<1\) and also precludes
a perfect MERP strategy.
\end{proof}

For a MERP strategy, question \(j\) of player \(\alpha\) is assigned the
observable \(M(\theta_{\alpha,j})\).  Identifying the output bit
\(a_\alpha\) with the eigenvalue \((-1)^{a_\alpha}\), clause \(c\) is won
with certainty exactly when
\begin{equation}\label{eq:merp-phase}
  \sum_{\alpha=1}^{4}
  \theta_{\alpha,q_c^{(\alpha)}}
  =\pi b_c\pmod{2\pi}.
\end{equation}
Thus perfect MERP satisfiability is an abelian phase-consistency problem.

\subsection{Abelian and noncommutative obstructions}
\label{sec:obstructions}

Let \(A\) be the clause-by-player-question incidence matrix of the reduced
support:
\[
  A_{e,(\alpha,j)}
  =
  \begin{cases}
    1,&q_e^{(\alpha)}=j,\\
    0,&q_e^{(\alpha)}\neq j.
  \end{cases}
\]
We call the block of columns associated with a fixed player, and by extension
the corresponding projection of a clause word, that player's \emph{page}.
Its integer relation lattice and its parity image are
\begin{align}
  K_E&=\ker_{\mathbb Z}(A^{\mathsf T}),\\
  L_E&=\{y\bmod2:y\in K_E\}\subseteq\mathbb F_2^E.
  \label{eq:LE-main}
\end{align}
An integer vector \(y\) is a \emph{parity-permuted refutation} (PREF) specification if
\begin{equation}\label{eq:pref-definition}
  A^{\mathsf T}y=0,
  \qquad
  b^{\mathsf T}y=1\pmod2.
\end{equation}
The noncommutative obstruction retains the order of the clauses.  A clause
word \(w=e_1\cdots e_N\) is \emph{balanced} if, for each player, the projected
word in the question involutions freely reduces to the identity.  Let
\(\pi(w)\in\mathbb F_2^E\) record the parity with which each clause occurs and
set
\begin{equation}\label{eq:RE-main}
  R_E=\operatorname{span}_{\mathbb F_2}
  \{\pi(w):w\text{ is balanced}\}.
\end{equation}
A balanced word with \(b^{\mathsf T}\pi(w)=1\) is a \emph{true refutation}.

For completeness, consider the standard game group generated by a central
involution \(J\) and question involutions \(x_{\alpha,j}\), with generators
belonging to distinct players required to commute.  Clause \(c\) is represented by
\[
  J^{b_c}\prod_{\alpha=1}^{4}x_{\alpha,q_c^{(\alpha)}}.
\]
An ordered product of these clause elements equals \(J\) exactly when the
corresponding clause word is a true refutation.  The following proposition
records the two proof interfaces used throughout the paper.

\begin{proposition}[MERP--PREF and strategy--refutation dualities]
\label{prop:obstruction-dualities}
For every finite reduced binary XOR game,
\begin{align}
  G\text{ has a perfect MERP strategy}
  &\quad\Longleftrightarrow\quad b\in L_E^\perp
  \quad\Longleftrightarrow\quad\text{no PREF exists},
  \label{eq:merp-main}\\
  \omega_{\mathrm{co}}(G)=1
  &\quad\Longleftrightarrow\quad b\in R_E^\perp
  \quad\Longleftrightarrow\quad\text{no true refutation exists}.
  \label{eq:co-main}
\end{align}
Consequently, \(\omega_{\mathrm{co}}(G)=1\) without a perfect MERP strategy
if and only if a PREF exists but no true refutation exists.
\end{proposition}

\begin{proof}
The first equivalence is the MERP--PREF alternative
\cite[Theorem~21 and ~22]{WattsHarrowKanwarNatarajan2019}.
For the second, the XOR game-group criterion gives the value-one equivalence
\cite[Theorem~2.1]{WattsHelton2023}; see also the
Nullstellensatz formulation in \cite[Sec.~3.5.1, Theorem~5.12, and Sec.~5.3.2]{WattsHeltonKlep2023}.  Orthogonality to
\(R_E\) is precisely the absence of a balanced word with odd target parity.
In the cited terminology, \(A^{\mathsf T}y=0\) is the parity-permuted
incidence condition, while a balanced word records an ordered product of
clause generators in the game group.
\end{proof}

\subsection{The Klein four-group witness}\label{sec:klein-witness}

Write \(V_4=\{0,a,b,c\}\), with \(a+b=c\), and index both the players and
their questions by \(V_4\).  The eight clauses of the separating game are
listed in Table~\ref{tab:klein}.

\begin{table}[htbp]
\centering
\caption{The Klein four-group game.  The question tuple is ordered by the
players \(0,a,b,c\); all clauses have weight \(1/8\), and the last column
gives the target parity.}
\label{tab:klein}
\begin{tabular}{ccc}
\toprule
Clause & Question tuple & Target\\
\midrule
\(E_0\) & \((0,0,0,0)\) & \(0\)\\
\(E_a\) & \((a,a,a,a)\) & \(0\)\\
\(E_b\) & \((b,b,b,b)\) & \(0\)\\
\(E_c\) & \((c,c,c,c)\) & \(0\)\\
\(O_0\) & \((0,a,b,c)\) & \(1\)\\
\(O_a\) & \((a,0,c,b)\) & \(0\)\\
\(O_b\) & \((b,c,0,a)\) & \(0\)\\
\(O_c\) & \((c,b,a,0)\) & \(0\)\\
\bottomrule
\end{tabular}
\end{table}

Every player-question pair occurs once in the \(E\)-family and once in the
\(O\)-family.  The phase contradiction is therefore immediate: summing the
four \(O\)-clause equations in Eq.~\eqref{eq:merp-phase} and subtracting the
four \(E\)-clause equations cancels every local angle, whereas the target
side is \(\pi\) modulo \(2\pi\), because \(O_0\) is the unique odd-target
clause.  Equivalently, the coefficient assignment
\begin{equation}\label{eq:cayley-pref}
  E_g\longmapsto-1,\qquad O_g\longmapsto1\qquad(g\in V_4)
\end{equation}
is an integer incidence relation of odd target parity, hence a PREF.
Theorem~\ref{thm:all-length-obstruction} proves that no true refutation
exists, irrespective of length, clause order, or multiplicities, for every
elementary abelian \(2\)-group; the Klein game is the rank-two case.  Together
with Eqs.~\eqref{eq:merp-main} and \eqref{eq:co-main}, this yields part~(ii)
of Theorem~\ref{thm:main}.

\section{Exact circuit lifting for four ternary pages}
\label{sec:sufficiency}

Using the two dualities in Proposition~\ref{prop:obstruction-dualities},
Theorem~\ref{thm:main}(i) reduces to the exactness of abelianization for four
pages with at most three active questions, namely
\begin{equation}\label{eq:central-target}
  L_E=R_E
  \qquad\text{for every }E\subseteq[3]^4.
\end{equation}
The inclusion \(R_E\subseteq L_E\) holds in general.  The content of this
section is the reverse inclusion for four-partite ternary supports, which
requires an unordered integral incidence relation to be realized by a single
clause ordering whose four page projections all freely reduce.

We prove a statement that retains the full integral multiplicity data.  The
\emph{alternating vector} of a word \(w=e_1\cdots e_N\) is
\[
  y(w)_e=\sum_{\{t:e_t=e\}}(-1)^{t-1}.
\]
A \emph{primitive support-minimal circuit} is a nonzero vector
\(c\in\ker_{\mathbb Z}(A^{\mathsf T})\) such that no nonzero relation has
support strictly contained in \(\supp c\), and whose nonzero coordinates have
greatest common divisor one.  It is determined by its support up to sign.

\begin{theorem}[Exact circuit lifting]\label{thm:exact-circuit-lifting}
Let \(E\subseteq[3]^4\) be a distinct clause support, let \(A\) be its
clause-by-player-question incidence matrix, and let
\(c\in\ker_{\mathbb Z}(A^{\mathsf T})\) be a primitive support-minimal
circuit.  There exists a balanced word \(w_c\) that uses each clause \(e\)
exactly \(\lvert c_e\rvert\) times and satisfies
\begin{equation}\label{eq:exact-circuit-lift}
  y(w_c)=c\ \text{ or }\ -c.
\end{equation}
In particular, \(\pi(w_c)=c\bmod2\).
\end{theorem}

\begin{corollary}
\label{cor:obstruction-exactness}
For every distinct support \(E\subseteq[3]^4\), one has \(L_E=R_E\).
\end{corollary}

The exact clause multiplicities are stronger than a parity realization for
each primitive circuit.  This strengthened
conclusion makes the passage from integral circuit generation to the
refutation space immediate.  The proof has three ingredients.  Support minimality
localizes every circuit to at most ten clauses.  Signed-occurrence matchings
then convert the ordering problem into a linear-forest problem.  Finally,
ternary Hamming geometry excludes every configuration for which the required
forest might fail to exist.

\subsection{Algebraic localization}\label{sec:parity}

We record the algebraic reductions used to pass between the circuit-lifting
theorem and the obstruction spaces, and then localize primitive circuits to
bounded support.

\begin{lemma}\label{lem:alternating-lift}
For every balanced word \(w=e_1\cdots e_N\), its alternating vector
satisfies \(A^{\mathsf T}y(w)=0\) and
\(y(w)\bmod2=\pi(w)\).  Hence \(R_E\subseteq L_E\).
\end{lemma}

\begin{proof}
Fix a player and a complete free reduction of the corresponding projected
word.  Matched equal letters
occupy positions \(r<s\) whose intervening subword also reduces to the
empty word.  That subword has even length, so \(r\) and \(s\) have opposite
parity.  Their contributions cancel in the alternating sum for their common
question.  Summing over reduction pairs proves every coordinate of
\(A^{\mathsf T}y(w)=0\).  Modulo two, both signs equal \(1\), giving the
parity identity.
\end{proof}

Thus only \(L_E\subseteq R_E\) remains.  The first reduction is genuinely
integral: rational spanning of the relation space would not suffice for the
parity conclusion.

\begin{proposition}\label{prop:circuit-generation}
The lattice \(K_E\) is generated over \(\Z\) by its primitive
support-minimal circuits.
\end{proposition}

\begin{proof}
We use induction on \(\lvert\supp y\rvert\), for \(y\in K_E\).  The assertion
is immediate for \(y=0\).  Otherwise choose a nonzero relation whose support
is contained in \(\supp y\) and is minimal among such supports.  Dividing its
coefficients by their greatest common divisor produces a primitive circuit
\(c\) with \(\supp c\subseteq\supp y\).

For \(e\in\supp c\), the relation \(c_e y-y_e c\) vanishes in coordinate
\(e\).  If it is nonzero, its support is strictly smaller than \(\supp y\),
and hence it is an integer combination of primitive circuits by induction.
Thus \(c_e y=(c_e y-y_e c)+y_e c\) belongs to the integral circuit span in
every case.  Primitivity gives
\(\gcd\{c_e:e\in\supp c\}=1\), so B\'ezout's identity supplies integers
\(a_e\) with \(\sum_{e\in\supp c}a_e c_e=1\).  Therefore
\[
  y=\sum_{e\in\supp c}a_e(c_e y)
\]
also lies in the integer span of the primitive circuits.
\end{proof}

\begin{lemma}
\label{lem:no-singleton-block}
Let \(c\in K_E\) be a circuit with support \(S=\supp c\).  Every
player-question block that is active on \(S\) contains at least two clauses
of \(S\).
\end{lemma}

\begin{proof}
For every player \(\alpha\) and question \(j\), the coordinate equation
\((A^{\mathsf T}c)_{(\alpha,j)}=0\) is the sum of \(c_e\) over clauses in
that block.  If its intersection with \(S\) were the singleton \(\{e\}\),
the equation would give \(c_e=0\), contradicting \(e\in\supp c\).
\end{proof}

\begin{proposition}\label{prop:rank-cap}
If \(c\) is a primitive circuit on \(S\), then, for the restriction \(A_S\)
of \(A\) to the rows indexed by \(S\),
\[
  \rank_{\Q}A_S=\card S-1,
  \qquad
  \card S\leq
  2+\sum_{\alpha=1}^{4}
  \bigl(\card{\{q_e^{(\alpha)}:e\in S\}}-1\bigr)\leq10.
\]
At ten clauses all pages are ternary; at nine clauses the question-count
profile is, up to permutation, \((3,3,3,2)\) or \((3,3,3,3)\).
\end{proposition}

\begin{proof}
If the rational relation space on \(S\) had dimension at least two, choose
\(z\in\Q^S\) independent of \(c\), and clear denominators.  Since every
\(c_e\neq0\), some
\(c_e z-z_e c\) is nonzero and vanishes at \(e\), contradicting minimal
support after clearing denominators once more.  Thus the nullity is one and
the rank is \(\card S-1\).

For every page, its nonempty question-indicator columns sum to the same
all-ones vector.  Hence that vector, together with
one fewer contrast than the number of active questions on each page, spans
all incidence columns.  This proves both the rank bound and the asserted
question-count profiles.
\end{proof}

The rank cap shows that a primitive circuit has at most ten support clauses,
although the coefficients \(\lvert c_e\rvert\), and hence the length of an
exact lift, remain unrestricted.  It therefore remains to solve an ordering
problem on at most ten distinct clause types.  The next subsection constructs
such orderings from signed-occurrence forests.

\subsection{Exact lifts from signed-occurrence forests}\label{sec:forests}

We now isolate the noncommutative part of the proof.  Fix a primitive circuit
\(c\).  Its coefficients prescribe signed clause multiplicities but do not
provide a common ordering in which all four page projections freely reduce.  Replace clause
\(e\) by \(\lvert c_e\rvert\) occurrence vertices and put these vertices in
parts \(P\) and \(N\) according to the sign of \(c_e\).  Summing the incidence
equations over the questions of any one player gives \(\sum_e c_e=0\), so
\(P\) and \(N\) have the same size.  More precisely, every player-question
block contains equally many vertices from the two parts.

We call \(c\) \emph{liftable} if there is a balanced word using each clause
\(e\) exactly \(\lvert c_e\rvert\) times.  Such a word is an \emph{exact
lift}; the following lemma shows that its position signs necessarily recover
the circuit signs.

\begin{lemma}\label{lem:exact-multiplicity}
If a balanced word uses clause \(e\) exactly \(\lvert c_e\rvert\) times, then
its alternating vector is \(c\) or \(-c\).  Consequently, after replacing
\(c\) by \(-c\) if necessary, the positive and negative occurrences occupy
the odd and even positions, respectively.
\end{lemma}

\begin{proof}
By Lemma~\ref{lem:alternating-lift}, the alternating vector is a rational
relation on the circuit support.  The relation space on that support is
one-dimensional, so it equals \(\lambda c\).  Since \(c\) is primitive and
the alternating vector is integral, \(\lambda\in\Z\); the coordinate bounds
give \(\lvert\lambda\rvert\leq1\).  If \(\lambda=0\), then every count
\(\lvert c_e\rvert\) is even, because an occurrence count and its alternating
sum have the same parity.  This contradicts primitivity.  Hence
\(\lambda=\pm1\), and equality in every coordinate bound places all
occurrences of one circuit sign on one position parity.
\end{proof}

\begin{lemma}\label{lem:ordering}
(i) A word on at most two letters with zero alternating sum for each letter
freely reduces to the empty word.
(ii) A finite bipartite linear forest with bipartition \(P\sqcup N\) and
\(\card P=\card N\) has a total vertex order alternating between the parts and placing
every forest edge consecutively.
\end{lemma}

\begin{proof}
For (i), if the word has no adjacent equal letters, every nonempty reduced
word on two letters alternates the letters.  Its two letterwise alternating
sums cannot both vanish.  Thus an adjacent equal pair exists.  Delete it and
argue by induction; deleting two consecutive positions preserves the parity
of every remaining position.

For (ii), traverse each path component from one endpoint to the other.
A component of even order contains equally many vertices from \(P\) and
\(N\), and its traversal may start in either part.  A component of odd order
contains one extra vertex in the part containing its endpoints.  The equality
\(\card P=\card N\) pairs the odd-order components with opposite excesses.
Concatenating each such pair, and then the even-order components with suitable
orientations, gives an alternating total order.  Every path edge joins
consecutive vertices in this order.
\end{proof}

The matching problem is encoded by signatures.  Select at most one question
block on each ternary page and regard each selected block as a color.  An
occurrence vertex receives the set of colors of the selected blocks that
contain it.

\begin{definition}[Balanced signature system]
A balanced signature system consists of finite vertex sets \(P,N\), a finite
color set, and a signature (a subset of the colors) attached to each vertex,
such that every color occurs equally often in \(P\) and \(N\).  A colorwise
matching pairs positive and negative incidences of each color.  If two
endpoints share two colors, the same simple edge may serve both matchings;
the forest condition refers to the union of the resulting simple edges.
\end{definition}

\begin{proposition}\label{prop:selected-block}
If the selected question blocks admit positive-to-negative matchings whose
simple union is a bipartite linear forest, then \(c\) has an exact lift.
\end{proposition}

\begin{proof}
Order the occurrence vertices by Lemma~\ref{lem:ordering}(ii), with \(P\) and
\(N\) alternating.  Reading their clause labels gives a word with the
prescribed multiplicities.  Every matching edge joins consecutive positions.
Consequently, on a page with a selected question, all occurrences of that
question cancel in adjacent equal pairs.  The same simple edge may carry two
colors, in which case the corresponding pair cancels on both pages.

After these pairs are deleted from a fixed page projection, at most two
question letters remain.  Deletion takes place in pairs, so the position
parities of the remaining letters do not change.  Because the total order
alternates between \(P\) and \(N\), the alternating sum for every remaining
question equals its coordinate in \(A^{\mathsf T}c\), up to one common sign,
and is therefore zero.
Lemma~\ref{lem:ordering}(i) now freely reduces the residual page word.  This
holds on every page, so the clause word is balanced and is an exact lift.
\end{proof}

Thus exact circuit lifting has been reduced to a matching statement: select
question blocks and pair their signed occurrences so that the simple union
has maximum degree two and no cycle.  The next two propositions provide the
matching statements needed below.

\begin{proposition}\label{prop:three-color}
For at most three colors, every balanced signature system with signatures
of size at most two has colorwise matchings whose simple union is a
bipartite linear forest.
\end{proposition}

\begin{proof}
Proceed by induction on the total number of color incidences.  First match and
remove opposite-sign vertices with identical nonempty signatures, using one
simple edge for both colors of a two-color signature.  Record these isolated
edges and restore them after treating the residual system.  No nonempty
signature then occurs on both signs in that system.

If a singleton \(\{a\}\) remains, match its \(a\)-incidence to any
opposite-sign endpoint incident with \(a\), delete \(a\) from both signatures,
and apply induction.  The singleton endpoint is isolated in the residual
system; restoring the edge therefore creates neither a cycle nor degree
greater than two.

It remains to exclude an active residual system consisting only of
two-element signatures.  For \(\sigma\in\{ab,ac,bc\}\), let \(\delta_\sigma\)
be the multiplicity of signature \(\sigma\) in \(P\) minus its multiplicity
in \(N\).  Color balance says
\[
\begin{pmatrix}1&1&0\\1&0&1\\0&1&1\end{pmatrix}
\begin{pmatrix}
\delta_{ab}\\\delta_{ac}\\\delta_{bc}
\end{pmatrix}
=0.
\]
The matrix has determinant \(-2\), so every \(\delta_\sigma=0\).  Since no
nonempty signature occurs on both signs in the residual system, all three
multiplicities must vanish, a contradiction.  The cases with fewer active
colors are contained in the same argument.
\end{proof}

\begin{proposition}
\label{prop:four-singleton}
For four colors, a balanced signature system of signature size at most two
has a linear-forest matching whenever either sign contains a singleton.
\end{proposition}

\begin{proof}
Induct on total color incidence.  Let a positive vertex have signature
\(\{a\}\).  If a negative endpoint has signature \(\{a,b\}\), delete
\(a\) from both endpoints.  The positive endpoint becomes empty and the
negative endpoint becomes the singleton \(\{b\}\); apply induction and
restore the edge incident with the now-isolated positive endpoint.

Otherwise every negative endpoint incident with \(a\) has signature
\(\{a\}\).  Delete one positive-negative singleton pair.  If another
singleton remains, recurse.  If not, no remaining negative endpoint uses
\(a\), and color balance says no positive endpoint uses it either.  The
residual system uses only three colors, so
Proposition~\ref{prop:three-color} applies.  Restore the isolated pair.
The negative-singleton case is symmetric.
\end{proof}

The support bound and the three-color forest already establish liftability
for every circuit that does not use three questions on all four pages.

\begin{proposition}
\label{prop:four-ternary-reduction}
Let \(c\) be a primitive circuit with \(\card{\supp c}\leq10\).  If at most
three pages use three questions on \(\supp c\), then \(c\) is liftable.
\end{proposition}

\begin{proof}
If at most two pages are ternary, select one question block on each of them.
There are at most two colors, so every signature has size at most two and
Propositions~\ref{prop:three-color} and \ref{prop:selected-block} apply.

Suppose exactly three pages are ternary.  The projection of at most ten
clauses onto these pages omits a point of \([3]^3\).  Select the three
question blocks defining such an omitted point.  No occurrence has all three
colors, so again every signature has size at most two.  The same two
propositions give an exact lift.  Pages using at most two questions require no
selected block and are handled by Lemma~\ref{lem:ordering}(i).
\end{proof}

\subsection{Hamming geometry of a nonliftable circuit}\label{sec:safe}

By Proposition~\ref{prop:four-ternary-reduction}, a nonliftable primitive
circuit must use all three questions on every page.  We identify its support
with a subset \(S\subset[3]^4\) of the ternary Hamming cube and use a point
outside the radius-one neighborhood of \(S\) to select one block on each
page.

Set
\[
 U(S)=\{u\in[3]^4:\dist(u,S)\geq2\},
\]
and for \(u\in U(S)\) write
\(a_u=\lvert\{q\in S:\dist(u,q)=2\}\rvert\).  We call the points of
\(U(S)\) \emph{safe centers}.  At a safe center \(u\), select the four
question blocks specified by the coordinates of \(u\).  The signature of a
clause \(q\) is its set of agreements with \(u\), and hence has size
\(4-\dist(u,q)\leq2\).

If \(\dist(u,q)=3\), this signature is a singleton.
Propositions~\ref{prop:four-singleton} and \ref{prop:selected-block} would
then lift the circuit.  Thus nonliftability forces the even-layer condition
\begin{equation}\label{eq:hard-geometry-main}
  \dist(u,q)\in\{2,4\}
  \qquad(u\in U(S),\ q\in S).
\end{equation}
A safe center is therefore a certificate of liftability unless the support
has this rigid even-layer form.

The packing estimates and the complete exclusion of the eight-, nine-, and
ten-point extremal configurations are given in
Appendix~\ref{app:extremal-hamming}.  Their conclusions are invoked below
only through Propositions~\ref{prop:subnine}, \ref{prop:nine},
\ref{prop:ten-noncover}, and \ref{prop:ten-cover}.

\begin{proof}[Proof of Theorem~\ref{thm:exact-circuit-lifting}]
Let \(c\) be a primitive circuit with support \(S\).
Proposition~\ref{prop:rank-cap} gives \(\card S\leq10\).  If at most three pages are ternary,
Proposition~\ref{prop:four-ternary-reduction} supplies an exact lift.  If all
four pages are ternary, Propositions~\ref{prop:subnine} and \ref{prop:nine}
handle support sizes at most nine.  For support size ten, either
\(U(S)\neq\varnothing\), in which case
Proposition~\ref{prop:ten-noncover} applies, or the ten radius-one balls cover
the cube, which Proposition~\ref{prop:ten-cover} shows is incompatible with a
primitive circuit.  Thus every primitive circuit has an exact lift.  Its
alternating vector is \(c\) or \(-c\) by
Lemma~\ref{lem:exact-multiplicity}, proving
Eq.~\eqref{eq:exact-circuit-lift}.
\end{proof}

\begin{proof}[Proof of Corollary~\ref{cor:obstruction-exactness}]
Proposition~\ref{prop:circuit-generation} writes every element of \(K_E\) as
an integer combination of primitive circuits.  Reducing this identity modulo
two shows that their parity vectors span \(L_E\) over \(\F_2\).  Each such
parity vector equals \(\pi(w_c)\) for a balanced word, and hence belongs to
\(R_E\).  Therefore \(L_E\subseteq R_E\).  The reverse inclusion is
Lemma~\ref{lem:alternating-lift}, so \(L_E=R_E\).  Notice that this identity
depends only on the support \(E\), not on the target vector \(b\).
\end{proof}

\begin{proof}[Proof of Theorem~\ref{thm:main}(i)]
If one question tuple occurs with both targets, the commuting-operator value
is strictly below one and no perfect MERP strategy exists, as observed in
Proposition~\ref{prop:weight-reduction}.  Otherwise reduce to a distinct
positive-weight support.  Proposition~\ref{prop:weight-reduction} permits the
uniform distribution, and the active-question hypothesis identifies the
support with a subset \(E\subseteq[3]^4\).
Corollary~\ref{cor:obstruction-exactness} gives \(L_E=R_E\).  The dualities in
Eqs.~\eqref{eq:merp-main} and
\eqref{eq:co-main} now yield
\[
  \omega_{\co}(G)=1
  \Longleftrightarrow b\in R_E^\perp
  \Longleftrightarrow b\in L_E^\perp
  \Longleftrightarrow G\text{ admits a perfect MERP strategy}.
\]
\end{proof}

\section{Sharpness at four questions}\label{sec:sharpness}

Part~(ii) of Theorem~\ref{thm:main} requires more than the PREF exhibited in
Subsection~\ref{sec:klein-witness}: one must also exclude every true refutation.
We separate this task into two logically distinct levels.  First, we classify
the shortest natural candidates---words using every Cayley clause exactly
once---for an arbitrary finite group.  We then specialize to elementary
abelian \(2\)-groups and use the Magnus expansion to rule out refutations with
arbitrary order and multiplicity.  The second step is essential; the first
alone cannot establish commuting-operator value one.

\subsection{Each-clause-once refutations in Cayley games}\label{sec:cayley}

Before addressing refutations of arbitrary length, we isolate the shortest
natural candidate: a word in which every Cayley clause occurs exactly once.
For each player, free reduction of such a word is equivalent to a
noncrossing matching of equal questions.  The group translations produce a
family of matchings that must be noncrossing in one common linear order.
Although simultaneous noncrossing matchings are related to partitioned book
embeddings \cite{AkitayaEtAl2017}, the regular group action makes this family
completely classifiable.

Let \(H\) be a finite group of order \(n\).  The associated Cayley game has
players and question labels
indexed by \(H\), clauses
\[
  E_g=(g)_{\alpha\in H},
  \qquad
  O_g=(g\alpha)_{\alpha\in H},
\]
and the unique odd target on \(O_1\).  For \(h\in H\), let \(m_h\) denote
the perfect matching between the two clause copies defined by
\(E_x\leftrightarrow O_{xh}\).

\begin{theorem}
\label{thm:cayley-cyclic}
The family \(\{m_h:h\in H\}\) admits a common noncrossing linear order if
and only if \(H\) is cyclic.  For \(n\geq2\), this is equivalent to the
existence of a true refutation that uses every clause of the associated
Cayley game exactly once.
\end{theorem}

This theorem explains the failure of an each-clause-once refutation for the
Klein group, but it is not by itself an all-length obstruction.  The latter
requires the Magnus argument of Subsection~\ref{sec:magnus}.

The classification proof is given in
Appendix~\ref{app:cayley-classification}.  Its noncrossing-matching
criterion is logically independent of the all-length Magnus obstruction
proved next.

\subsection{All-length obstruction for elementary abelian Cayley games}
\label{sec:magnus}

\begin{theorem}[All-length elementary-abelian obstruction]
\label{thm:all-length-obstruction}
Let \(H=\F_2^r\), \(r\geq2\), and consider the associated Cayley game with
the unique odd target on \(O_0\).  The coefficient assignment
\(E_g\mapsto-1\), \(O_g\mapsto1\) is a PREF, but no true refutation exists.
Consequently, the game has commuting-operator value one and no perfect MERP
strategy.
\end{theorem}

\begin{proof}
Assume that a true refutation of length
\(L\) exists.  Write the clause in position \(p\) as \((g_p,m_p)\), where
\(m_p=0\) for \(E_{g_p}\) and \(m_p=1\) for \(O_{g_p}\).  Player
\(\alpha\) receives question \(g_p+m_p\alpha\).
We first rewrite each player projection in a free basis for the even
subgroup.  Degree one of the Magnus expansion fixes all alternating clause
sums in terms of a single odd integer.  Degree two then supplies pair-count
identities, and an integral certificate contradicts that oddness.

We begin with the required normal form for the even subgroup.

Let \(F=\ast_{q\in H}\mathbb Z_2\) be generated by question involutions
\(x_q\), and let \(F^{\mathrm e}\) be its even subgroup.

\begin{lemma}\label{lem:even-subgroup}
For \(q\ne0\), set \(t_q=x_qx_0\), and put \(t_0=1\).  The elements
\(\{t_q:q\ne0\}\) freely generate \(F^{\mathrm e}\).  Every even word
satisfies
\[
  x_{a_1}\cdots x_{a_{2N}}
  =
  \prod_{i=1}^N t_{a_{2i-1}}t_{a_{2i}}^{-1}.
\]
\end{lemma}

\begin{proof}
The homomorphism from \(F\) to \(\mathbb Z_2\) that sends every \(x_q\) to
the nontrivial element has kernel \(F^{\mathrm e}\), with Schreier transversal
\(\{1,x_0\}\).  The Reidemeister--Schreier generators associated with
\(x_q\), \(q\ne0\), are
\[
  x_qx_0=t_q,
  \qquad
  x_0x_q=t_q^{-1}.
\]
The generator associated with \(x_0\) is trivial, and rewriting the
relations \(x_q^2=1\) produces only the tautological relations
\(t_qt_q^{-1}=1\).  The subgroup is therefore free on the \(t_q\)'s
\cite{MagnusKarrassSolitar1976}.  Finally,
\(t_qt_s^{-1}=x_qx_s\); multiplying this identity over consecutive pairs
proves the displayed normal form.
\end{proof}

Consequently, the projection of player \(\alpha\) reduces to the identity
if and only if
\begin{equation}\label{eq:player-normal-form}
  \prod_{i=1}^{L/2}
  X_{g_{2i-1}+m_{2i-1}\alpha}
  X_{g_{2i}+m_{2i}\alpha}^{-1}
  =1
\end{equation}
in the free group generated by \(X_q\), \(q\ne0\), with \(X_0=1\).

We use the degree-one and class-two terms of the Magnus expansion
\cite{Magnus1937,MagnusKarrassSolitar1976}.

\begin{lemma}\label{lem:magnus-conditions}
Let \(W=X_{a_1}^{\varepsilon_1}\cdots
X_{a_N}^{\varepsilon_N}\) be a word in a free group, with
\(\varepsilon_p\in\{\pm1\}\), and let \([P]\) denote the indicator of a
statement \(P\).  If \(W=1\), then
\begin{align}
  \sum_p\varepsilon_p[a_p=i]&=0,
  \label{eq:magnus-one}\\
  \sum_{p<q}\varepsilon_p\varepsilon_q
  [a_p=i][a_q=j]&=0
  \qquad(i\ne j).
  \label{eq:magnus-two}
\end{align}
\end{lemma}

\begin{proof}
Apply the Magnus homomorphism into noncommutative formal power series,
truncated above degree two:
\[
  X_i\longmapsto1+u_i,
  \qquad
  X_i^{-1}\longmapsto1-u_i+u_i^2.
\]
The coefficient of \(u_i\) is the left-hand side of
Eq.~\eqref{eq:magnus-one}.  For \(i\ne j\), the coefficient of the ordered
monomial \(u_i u_j\) receives no contribution from the quadratic term of a
single inverse and is therefore exactly the left-hand side of
Eq.~\eqref{eq:magnus-two}.  If \(W=1\), every positive-degree coefficient
vanishes, proving both identities.
\end{proof}

We first extract the degree-one consequences.

Let \(s_p=1\) for odd \(p\) and \(s_p=-1\) for even \(p\), and define
\[
  V_g=\sum_{\substack{p:g_p=g\\m_p=0}}s_p,
  \qquad
  W_g=\sum_{\substack{p:g_p=g\\m_p=1}}s_p.
\]
The exponent of \(X_j\) in the word on the left of
Eq.~\eqref{eq:player-normal-form} is
\[
  V_j+W_{j+\alpha}.
\]
Eq.~\eqref{eq:magnus-one}, for \(j\ne0\), therefore implies
\[
  W_g=-v\quad(g\in H),
  \qquad
  V_j=v\quad(j\ne0)
\]
for one integer \(v\).  Since \(L\) is even,
\(\sum_p s_p=0\), and summing all \(V_g\) and \(W_g\) gives \(V_0=v\).

Only \(O_0\) has odd target parity.  Modulo two,
\[
  W_0
  =
  \sum_{p:O_0\text{ occurs at }p}s_p
  \equiv
  \#\{p:O_0\text{ occurs at }p\}.
\]
The parity condition of a true refutation thus forces \(v\) to be odd.

We now turn to the degree-two constraints.

All indices in what follows lie in \(H\).  For distinct nonzero \(i,j\),
define \(R_{ij}\); for \(i\ne0\) and arbitrary \(g\), define \(P_{ig}\);
for arbitrary \(g\) and \(j\ne0\), define \(Q_{gj}\); and for \(g\ne h\),
define \(S_{gh}\) by the signed pair counts
\begin{align*}
  R_{ij}
  &=
  \sum_{\substack{p<q\\E_i\text{ at }p,\ E_j\text{ at }q}}s_ps_q,
  &
  P_{ig}
  &=
  \sum_{\substack{p<q\\E_i\text{ at }p,\ O_g\text{ at }q}}s_ps_q,\\
  Q_{gj}
  &=
  \sum_{\substack{p<q\\O_g\text{ at }p,\ E_j\text{ at }q}}s_ps_q,
  &
  S_{gh}
  &=
  \sum_{\substack{p<q\\O_g\text{ at }p,\ O_h\text{ at }q}}s_ps_q.
\end{align*}
For fixed nonzero \(i\ne j\), the four possible combinations of clause
families contribute the terms \(R,P,Q,S\), respectively.  Applying
Eq.~\eqref{eq:magnus-two} to player \(c\) gives the system
\begin{equation}\label{eq:Eprime-app}
  P_{i,j+c}+Q_{i+c,j}+S_{i+c,j+c}=-R_{ij}
  \qquad(i,j\in H\setminus\{0\},\ i\ne j,\ c\in H).
\end{equation}
The constant term \(R_{ij}\) cannot be separated from the point masses in
this equation because the constant function is their sum.

Splitting products of the degree-one sums into the regions \(p<q\) and
\(q<p\) gives
\begin{align}
  R_{ij}+R_{ji}&=v^2 &&(i\ne j),\label{eq:VV-app}\\
  P_{ig}+Q_{gi}&=-v^2 &&(i\ne0),\label{eq:PW-app}\\
  S_{gh}+S_{hg}&=v^2 &&(g\ne h).\label{eq:WW-app}
\end{align}
For example, the first identity is the equality
\(V_iV_j=v^2\) decomposed according to the order of the two occurrences;
the other two are \(V_iW_g=-v^2\) and \(W_gW_h=v^2\).

It remains to combine these identities into a parity contradiction.

Identify a four-element subgroup of \(H\) with
\(\{0,1,2,3\}\cong V_4\), using bitwise XOR for its addition.  Every equation
whose indices lie in this subgroup is a member of the full system.
For the six ordered pairs \(i\ne j\) in \(\{1,2,3\}\), take the four
equations \eqref{eq:Eprime-app}, indexed by \(c=0,1,2,3\), with the
coefficients
\[
\begin{array}{c|rrrr}
(i,j)&c=0&c=1&c=2&c=3\\ \hline
(1,2)& 5&-3& 5&-3\\
(1,3)& 3& 3&-5&-5\\
(2,1)& 3&-5& 3&-5\\
(2,3)&-3& 5&-3& 5\\
(3,1)&-3& 5&-3&-3\\
(3,2)& 3& 3&-5& 3
\end{array}
\]
To this weighted sum, add four times Eq.~\eqref{eq:VV-app} for
\((i,j)=(1,2)\) and \((1,3)\), subtract four times its \((2,3)\) instance,
add eight times Eq.~\eqref{eq:PW-app} for \((i,g)=(1,1)\), and subtract eight
times its \((1,2)\) instance.  Collection by variable family gives
\[
\begin{array}{c|l}
\text{family}&\text{surviving contribution}\\
\hline
R&8R_{12}\\
P&0\\
Q&8(Q_{13}-Q_{23}-Q_{31}+Q_{32})\\
S&8(-S_{31}+S_{32})\\
\text{right-hand side}&4v^2
\end{array}
\]
and therefore
\[
  8\bigl(
  R_{12}+Q_{13}-Q_{23}-Q_{31}+Q_{32}-S_{31}+S_{32}
  \bigr)=4v^2.
\]
The expression in parentheses is an integer, so \(v^2\) is even.  This
contradicts the oddness of \(v\) and proves that no true refutation exists.

Finally, the assignment \(E_g\mapsto-1\), \(O_g\mapsto1\) is a PREF because
every player-question incidence occurs once in each clause family, whereas
only \(O_0\) has odd target parity.  Proposition
\ref{prop:obstruction-dualities} gives the two strategy consequences.  The
argument is independent of the length, order, and multiplicities of the
putative refutation.
\end{proof}

\begin{proof}[Proof of Theorem~\ref{thm:main}(ii)]
Take \(H=\F_2^2\cong V_4\).  Table~\ref{tab:klein} gives four players and
four active questions per player.  Theorem~\ref{thm:all-length-obstruction}
gives commuting-operator value one and excludes every perfect MERP strategy.
Part~(i) rules out such a separation with three or fewer active questions,
so the four-question construction is sharp.
\end{proof}

\section{Discussion}\label{sec:discussion}

Theorem~\ref{thm:main} determines the sharp local-question threshold for the
existence of a perfect GHZ-equatorial realization of four-player XOR-game
constraints.  If a game with at most three active questions per player has
commuting-operator value one, then its winning constraints also admit a
perfect realization on the four-qubit GHZ state with equatorial measurements.
At four active questions, the Klein game has commuting-operator value one but
admits no such realization.  This is an existential statement about realizing
a game's perfect winning constraints; it does not identify every perfect
commuting-operator strategy or full correlation with a GHZ-equatorial one.
The positive statement follows from exactness of abelianization on every
ternary four-partite support, while the Klein construction separates abelian
phase inconsistency from noncommutative refutability.  Both conclusions
concern the support and target vector and are independent of the positive
clause weights.

The proof also isolates the mechanism at the boundary.  Primitive incidence
circuits on ternary pages admit exact balanced-word lifts, whereas the Klein
translation pattern has an immediate phase obstruction but no refutation of
any length.  The latter conclusion follows from an integral parity
certificate in the degree-two Magnus quotient and is not based on a bounded
enumeration.

\paragraph{Open problems.}

The first question is whether the Klein game admits a perfect
finite-dimensional strategy.  A positive answer would place the first escape
from the GHZ-phase model already in finite dimension; a negative answer would
separate finite-dimensional tensor-product from commuting-operator
perfectness.

A second problem is to classify the supports \(E\subseteq[4]^4\) for which
\(R_E\subsetneq L_E\).  In particular, what is the minimum number of clauses,
and are all minimal examples equivalent, up to relabeling, to a Klein-type
translation pattern?

One may also ask for a quantitative version below the threshold.  If every
player has at most three active questions and
\(\omega_{\co}(G)\geq1-\varepsilon\), it is unclear whether there exist a GHZ-equatorial
strategy of value at least \(1-f(\varepsilon)\) for an explicit rate
\(f(\varepsilon)\to0\), preferably uniform over the clause weights. An robust game-algebra
methods may provide a useful framework \cite{Zhao2024}.

Finally, the all-length obstruction should be extended beyond elementary
abelian \(2\)-groups.  For which finite groups can repeated Cayley clauses form
a true refutation, and when do degree-two Magnus coefficients suffice to rule
this out, remains open.

\section*{Acknowledgments}
This work was supported by the National Natural Science Foundation of China
(grant nos. 92576114 and 12447107); Guangdong Provincial Quantum Science
Strategic Initiative (grant nos. GDZX2403008, GDZX2503001, and GDZX2403001).

\paragraph{Use of AI-assisted tools.}
Exploratory computation and the generalization of proofs in the present work were supported by a research harness system currently under development by QudeLeap Research and powered by official version of DeepSeek-V4-Flash~\cite{deepseek2026deepseekv4}. The authors independently formulated the research questions, established the mathematical conventions, and rigorously verified every statement and proof included in the manuscript. The authors take full responsibility for the final text and all mathematical claims contained therein.

\bibliographystyle{alpha}
\bibliography{references}

@article{Bell1964,
  author  = {Bell, John S.},
  title   = {On the {Einstein--Podolsky--Rosen} Paradox},
  journal = {Physics Physique Fizika},
  volume  = {1},
  number  = {3},
  pages   = {195--200},
  year    = {1964},
  doi     = {10.1103/PhysicsPhysiqueFizika.1.195}
}

@article{ClauserHorneShimonyHolt1969,
  author  = {Clauser, John F. and Horne, Michael A. and Shimony, Abner and Holt, Richard A.},
  title   = {Proposed Experiment to Test Local Hidden-Variable Theories},
  journal = {Physical Review Letters},
  volume  = {23},
  number  = {15},
  pages   = {880--884},
  year    = {1969},
  doi     = {10.1103/PhysRevLett.23.880}
}

@article{Mermin1990,
  author  = {Mermin, N. David},
  title   = {Extreme Quantum Entanglement in a Superposition of Macroscopically Distinct States},
  journal = {Physical Review Letters},
  volume  = {65},
  number  = {15},
  pages   = {1838--1840},
  year    = {1990},
  doi     = {10.1103/PhysRevLett.65.1838}
}

@article{GreenbergerHorneShimonyZeilinger1990,
  author  = {Greenberger, Daniel M. and Horne, Michael A. and Shimony, Abner and Zeilinger, Anton},
  title   = {Bell's Theorem without Inequalities},
  journal = {American Journal of Physics},
  volume  = {58},
  number  = {12},
  pages   = {1131--1143},
  year    = {1990},
  doi     = {10.1119/1.16243}
}

@article{WernerWolf2001,
  author  = {Werner, Reinhard F. and Wolf, Michael M.},
  title   = {All-Multipartite Bell-Correlation Inequalities for Two Dichotomic Observables per Site},
  journal = {Physical Review A},
  volume  = {64},
  number  = {3},
  pages   = {032112},
  year    = {2001},
  doi     = {10.1103/PhysRevA.64.032112}
}

@inproceedings{CleveHoyerTonerWatrous2004,
  author    = {Cleve, Richard and H{\o}yer, Peter and Toner, Benjamin and Watrous, John},
  title     = {Consequences and Limits of Nonlocal Strategies},
  booktitle = {Proceedings of the 19th {IEEE} Conference on Computational Complexity},
  pages     = {236--249},
  year      = {2004},
  doi       = {10.1109/CCC.2004.9},
  eprint    = {quant-ph/0404076},
  archivePrefix = {arXiv}
}

@article{PerezGarciaEtAl2008,
  author  = {P{\'e}rez-Garc{\'i}a, David and Wolf, Michael M. and Palazuelos, Carlos and Villanueva, Ignacio and Junge, Marius},
  title   = {Unbounded Violation of Tripartite Bell Inequalities},
  journal = {Communications in Mathematical Physics},
  volume  = {279},
  number  = {2},
  pages   = {455--486},
  year    = {2008},
  doi     = {10.1007/s00220-008-0418-4}
}

@article{BrietVidick2013,
  author  = {Bri{\"e}t, Jop and Vidick, Thomas},
  title   = {Explicit Lower and Upper Bounds on the Entangled Value of Multiplayer {XOR} Games},
  journal = {Communications in Mathematical Physics},
  volume  = {321},
  number  = {1},
  pages   = {181--207},
  year    = {2013},
  doi     = {10.1007/s00220-012-1642-5}
}

@article{NavascuesPironioAcin2008,
  author  = {Navascu{\'e}s, Miguel and Pironio, Stefano and Ac{\'i}n, Antonio},
  title   = {A Convergent Hierarchy of Semidefinite Programs Characterizing the Set of Quantum Correlations},
  journal = {New Journal of Physics},
  volume  = {10},
  number  = {7},
  pages   = {073013},
  year    = {2008},
  doi     = {10.1088/1367-2630/10/7/073013}
}

@article{CleveLiuSlofstra2017,
  author  = {Cleve, Richard and Liu, Li and Slofstra, William},
  title   = {Perfect Commuting-Operator Strategies for Linear System Games},
  journal = {Journal of Mathematical Physics},
  volume  = {58},
  number  = {1},
  pages   = {012202},
  year    = {2017},
  doi     = {10.1063/1.4973422}
}

@article{Slofstra2019,
  author  = {Slofstra, William},
  title   = {The Set of Quantum Correlations Is Not Closed},
  journal = {Forum of Mathematics, Pi},
  volume  = {7},
  pages   = {e1},
  year    = {2019},
  doi     = {10.1017/fmp.2018.3}
}

@article{JiNatarajanVidickWrightYuen2020,
  author  = {Ji, Zhengfeng and Natarajan, Anand and Vidick, Thomas and Wright, John and Yuen, Henry},
  title   = {{MIP}$^*={RE}$},
  journal = {Communications of the ACM},
  volume  = {64},
  number  = {11},
  pages   = {131--138},
  year    = {2021},
  doi     = {10.1145/3485628},
  eprint  = {2001.04383},
  archivePrefix = {arXiv}
}

@inproceedings{WattsHarrowKanwarNatarajan2019,
  author    = {Watts, Adam Bene and Harrow, Aram W. and Kanwar, Gurtej and Natarajan, Anand},
  title     = {Algorithms, Bounds, and Strategies for Entangled {XOR} Games},
  booktitle = {10th Innovations in Theoretical Computer Science Conference},
  series    = {Leibniz International Proceedings in Informatics},
  volume    = {124},
  pages     = {10:1--10:18},
  year      = {2019},
  doi       = {10.4230/LIPIcs.ITCS.2019.10},
  eprint    = {1801.00821},
  archivePrefix = {arXiv},
  url       = {https://arxiv.org/abs/1801.00821}
}

@article{WattsHelton2023,
  author  = {Watts, Adam Bene and Helton, J. William},
  title   = {{3XOR} Games with Perfect Commuting Operator Strategies Have Perfect Tensor Product Strategies and Are Decidable in Polynomial Time},
  journal = {Communications in Mathematical Physics},
  volume  = {400},
  pages   = {731--791},
  year    = {2023},
  doi     = {10.1007/s00220-022-04615-3}
}

@article{WattsHeltonKlep2023,
  author  = {Watts, Adam Bene and Helton, J. William and Klep, Igor},
  title   = {Noncommutative {Nullstellens{\"a}tze} and Perfect Games},
  journal = {Annales Henri Poincar{\'e}},
  volume  = {24},
  pages   = {2183--2239},
  year    = {2023},
  doi     = {10.1007/s00023-022-01262-1}
}

@misc{Zhao2024,
  author        = {Zhao, Yuming},
  title         = {Robust Self-Testing for Nonlocal Games with Robust Game Algebras},
  year          = {2024},
  eprint        = {2411.03259},
  archivePrefix = {arXiv},
  primaryClass  = {quant-ph},
  url           = {https://arxiv.org/abs/2411.03259}
}

@article{Magnus1937,
  author  = {Magnus, Wilhelm},
  title   = {{\"U}ber Beziehungen zwischen h{\"o}heren Kommutatoren},
  journal = {Journal f{\"u}r die reine und angewandte Mathematik},
  volume  = {177},
  pages   = {105--115},
  year    = {1937},
  doi     = {10.1515/crll.1937.177.105}
}

@book{MagnusKarrassSolitar1976,
  author    = {Magnus, Wilhelm and Karrass, Abraham and Solitar, Donald},
  title     = {Combinatorial Group Theory: Presentations of Groups in Terms of Generators and Relations},
  publisher = {Dover Publications},
  address   = {New York},
  year      = {1976}
}

@inproceedings{AkitayaEtAl2017,
  author    = {Akitaya, Hugo A. and Demaine, Erik D. and Hesterberg, Adam and Liu, Quanquan C.},
  title     = {Upward Partitioned Book Embeddings},
  booktitle = {Graph Drawing and Network Visualization},
  series    = {Lecture Notes in Computer Science},
  volume    = {10692},
  pages     = {210--223},
  publisher = {Springer},
  year      = {2018},
  doi       = {10.1007/978-3-319-73915-1_18}
}

@inproceedings{deMouraUllrich2021,
  author    = {de Moura, Leonardo and Ullrich, Sebastian},
  title     = {The {Lean 4} Theorem Prover and Programming Language},
  booktitle = {Automated Deduction---{CADE} 28},
  series    = {Lecture Notes in Computer Science},
  volume    = {12699},
  pages     = {625--635},
  publisher = {Springer},
  year      = {2021},
  doi       = {10.1007/978-3-030-79876-5_37}
}

@inproceedings{mathlib2020,
  author    = {{The mathlib Community}},
  title     = {The {Lean} Mathematical Library},
  booktitle = {Proceedings of the 9th {ACM SIGPLAN} International Conference on Certified Programs and Proofs},
  pages     = {367--381},
  publisher = {Association for Computing Machinery},
  year      = {2020},
  doi       = {10.1145/3372885.3373824}
}

@misc{ZhuEtAl2026LeanQIT,
  author        = {Zhu, Chengkai and Tang, Ziao and Zhen, Guocheng and Cao, Yimeng and Zhao, Yusheng and Chen, Ranyiliu and Zhao, Xuanqiang and Zhang, Lei and Wang, Xin},
  title         = {{Lean-QIT}: Towards a Formal Infrastructure for Quantum Information Theory},
  year          = {2026},
  eprint        = {2607.09632},
  archivePrefix = {arXiv},
  primaryClass  = {quant-ph},
  url           = {https://arxiv.org/abs/2607.09632}
}

@misc{FourPlayerXORGames2026,
  author       = {Tang, Ziao and Zhu, Chengkai and Bai, Ge and Wang, Xin and Chen, Ranyiliu},
  title        = {Formalization and Exact Certificates for a Sharp Local-Question Threshold in Four-Player {XOR} Games},
  year         = {2026. URL: https://github.com/QuAIR/Four-Player-XOR-Games},
  howpublished = {GitHub repository},
  url          = {https://github.com/QuAIR/Four-Player-XOR-Games}
}

@article{deepseek2026deepseekv4,
  title   = {DeepSeek-V4: Towards Highly Efficient Million-Token Context Intelligence},
  author  = {{DeepSeek-AI} and others},
  journal = {arXiv preprint arXiv:2606.19348},
  year    = {2026},
  doi     = {10.48550/arXiv.2606.19348}
}

\appendix

\begin{center}
\Large{Appendix}
\end{center}

\section{Extremal configurations in the ternary Hamming cube}
\label{app:extremal-hamming}

This appendix supplies the packing estimates and extremal exclusions used
in the proof of Theorem~\ref{thm:exact-circuit-lifting}.  The notation
\(S\), \(U(S)\), and \(a_u\) is that of Subsection~\ref{sec:safe}.

\begin{lemma}\label{lem:separation-packing}
For a nonliftable circuit with support \(S\subset[3]^4\), using all three
questions on each page and having at most ten clauses,
distinct safe centers have distance at least three and
\(\card{U(S)}\leq6\).
\end{lemma}

\begin{proof}
Suppose first that safe centers \(u,u'\) differ only on page \(a\), and let
\(r_a\) be the third question value on that page.  Comparing the parities in
Eq.~\eqref{eq:hard-geometry-main} for \(u\) and \(u'\) forces every clause to
use \(r_a\), contradicting the assumption that page \(a\) is ternary.

If they differ on pages \(a,b\), with third values \(r_a,r_b\), parity
comparison gives
\[
  \one_{\{q^{(a)}\neq r_a\}}
  =
  \one_{\{q^{(b)}\neq r_b\}}
\]
for every clause.  Thus the complete \(r_a\)-block on page \(a\) equals the
complete \(r_b\)-block on page \(b\).  On either page, the two other active
question blocks have size at least two by
Lemma~\ref{lem:no-singleton-block}.  Since \(\card S\leq10\), the common block
has size at most six.  Its projection onto the remaining two pages omits one
of their nine joint cells.  Regard the common occurrence block as one color
and select the two question blocks defining the omitted cell as the other
colors.  No occurrence has all three colors, so
Proposition~\ref{prop:three-color} supplies a linear forest.  The matching of
the common block can be used simultaneously on pages \(a\) and \(b\), because
the two selected blocks are identical as multisets of signed occurrences.
Proposition~\ref{prop:selected-block} then gives an exact lift.  This
contradiction proves that distinct safe centers are separated by distance at
least three.

For the packing bound, fix a clause and relabel it as \(0\).  By
Eq.~\eqref{eq:hard-geometry-main}, every safe center lies in Hamming layer two
or four about this clause.  The support coordinates of a layer-two center
define an edge on the four coordinate positions.  Two centers cannot define
the same edge, since their mutual distance would be at most two.  Moreover,
two incident edges must carry different nonzero labels at their common
coordinate; otherwise the corresponding centers again have distance two.
There are only two nonzero labels, so this simple graph has maximum degree two
and hence at most four edges.

Layer-four centers belong to \(\{1,2\}^4\).  This binary four-cube contains no
three points at mutual distance at least three: relative to one point, the
change sets of two others both have size at least three and hence have
symmetric difference at most two.  There are therefore at most two
layer-four centers, and \(\card{U(S)}\leq4+2=6\).
\end{proof}

\begin{lemma}\label{lem:safe-excess}
Let \(c\) be a nonliftable primitive circuit with support \(S\subset[3]^4\),
using three questions on all four pages and satisfying \(\card S\leq10\).
The total excess multiplicity of the radius-one balls centered at \(S\)---the
sum, over covered points, of their multiplicity minus one---is
\[
  9\card S-81+\card{U(S)}.
\]
Moreover,
\begin{align}
  a_u&\geq4\qquad(u\in U(S)),\label{eq:au-four-app}\\
  2\sum_{u\in U(S)}(a_u-4)
  &\leq9\card S-81+\card{U(S)}.\label{eq:excess-bound-app}
\end{align}
\end{lemma}

\begin{proof}
There are \(9\card S\) incidences between the radius-one balls centered at
\(S\) and the \(81\)-point cube.  Exactly \(81-\card{U(S)}\) points are
covered; subtracting one from every covered multiplicity gives the stated
total excess.

Fix \(u\in U(S)\).  Each of its eight neighbors is covered: otherwise that
neighbor would also belong to \(U(S)\), contrary to the separation in
Lemma~\ref{lem:separation-packing}.  The radius-one neighborhoods of distinct
safe centers are disjoint for the same reason.  By
Eq.~\eqref{eq:hard-geometry-main}, a support clause is at distance two or four
from \(u\).  A distance-two clause covers exactly two of the eight neighbors
of \(u\), whereas a distance-four clause covers none.  Hence the excess over
these eight covered points is
\[
  2a_u-8=2(a_u-4)\geq0,
\]
which proves \(a_u\geq4\).  The disjoint safe-center neighborhoods contribute
disjoint parts of the total excess computed above; summing their local
excesses proves Eq.~\eqref{eq:excess-bound-app}.
\end{proof}

\subsection{Exclusion of the extremal configurations}\label{sec:windows}

Only circuits on four ternary pages remain.  By
Proposition~\ref{prop:rank-cap}, their support size is at most ten.  We now
exclude the possible nonliftable configurations.  The Hamming-ball excess is
already decisive up to eight points.  At nine points, a covering configuration
would be a perfect packing.  At ten points, the noncovering and covering cases
are governed, respectively, by the layer-two fibers of a safe center and by
the cosets of the ternary Hamming code.

\begin{proposition}\label{prop:subnine}
Every primitive circuit on at most eight distinct clauses lifts.
\end{proposition}

\begin{proof}
Circuits with at most three ternary pages lift by
Proposition~\ref{prop:four-ternary-reduction}.  Suppose that all four pages are
ternary and, towards a contradiction, that a circuit with support \(S\) is
nonliftable.  Lemmas~\ref{lem:separation-packing} and
\ref{lem:safe-excess} give
\[
  0
  \leq
  2\sum_{u\in U(S)}(a_u-4)
  \leq
  9\card S-81+\card{U(S)}
  \leq72-81+6=-3.
\]
The same calculation includes the case \(U(S)=\varnothing\), for which the
left side is zero.
\end{proof}

We next treat nine-point supports.

\begin{lemma}\label{lem:nine-independence}
If nine radius-one balls centered at distinct points of \([3]^4\) cover the
cube, the clause-question incidence matrix \(A_S\) of their centers has
\(\rank_{\Q}A_S=9\).
\end{lemma}

\begin{proof}
The nine balls have total size \(81\), so a cover makes them disjoint and
their centers pairwise distance at least three.  Projection of the centers
onto any two coordinates is injective and hence bijective onto \([3]^2\).
It follows that, on the nine centers, the two-dimensional zero-mean contrast
spaces from distinct coordinates are mutually orthogonal; each is also
orthogonal to the constants.  The constant direction together with the four
contrast spaces therefore has dimension \(1+4\cdot2=9\) in the incidence
column span.  Hence the row rank is nine.
\end{proof}

\begin{proposition}\label{prop:nine}
Every nine-clause primitive circuit lifts.
\end{proposition}

\begin{proof}
Let \(S\) be the circuit support.
By Proposition~\ref{prop:rank-cap}, the question-count profile is, up to
permutation, \((3,3,3,2)\) or \((3,3,3,3)\).  The first case lifts by
Proposition~\ref{prop:four-ternary-reduction}.  In the second case, assume
nonliftability.  If \(U(S)=\varnothing\), the nine radius-one balls cover the
cube; Lemma~\ref{lem:nine-independence} gives \(\rank A_S=9\), whereas a
nine-point circuit has rank eight by Proposition~\ref{prop:rank-cap}.  Hence
\(U(S)\neq\varnothing\).

For nine support points, Eq.~\eqref{eq:excess-bound-app} becomes
\[
  2\sum_{u\in U(S)}(a_u-4)\leq\card{U(S)}.
\]
Since every \(a_u\geq4\), some safe center satisfies \(a_u=4\); otherwise the
left side would be at least \(2\card{U(S)}\).

The Hamming layer two about \(u\) has
\(\binom42 2^2=24\) points.  The ball of a distance-two clause meets this
layer in exactly three points, whereas a distance-four ball misses it.
The four distance-two clauses therefore cover at most twelve layer-two
points.  Every remaining point in that layer is uncovered and hence is a safe
center.  This produces at least twelve safe centers in addition to \(u\),
contrary to Lemma~\ref{lem:separation-packing}.
\end{proof}

It remains to consider ten-point supports.  We first treat the noncovering
case, in which a safe center exists.

\begin{proposition}
\label{prop:ten-noncover}
A ten-clause primitive circuit with support \(S\) and
\(U(S)\neq\varnothing\) lifts.
\end{proposition}

\begin{proof}
Assume nonliftability.  Lemma~\ref{lem:separation-packing} gives
\(1\leq\card{U(S)}\leq6\).  If \(a_u\leq6\) for some \(u\), the radius-one
balls of the distance-two clauses cover at most \(3a_u\) of the \(24\) points in layer
two about \(u\); the distance-four balls cover none.  Thus at least
\(24-3a_u\geq6\) further points belong to \(U(S)\), which together with \(u\)
contradicts \(\card{U(S)}\leq6\).  Hence \(a_u\geq7\) for every \(u\in U(S)\).
Equation~\eqref{eq:excess-bound-app} now gives
\[
  6\card{U(S)}\leq9+\card{U(S)},
\]
so \(\card{U(S)}=1\).

Relabel the unique center as \(0\).  Its 24 layer-two points split into six
fibers according to their two-coordinate support, with four choices of
nonzero labels in each fiber.  Only a distance-two clause with the same
coordinate support can cover a point in a given fiber, and the ball of one
such clause covers three of its four points.  Since \(0\) is the unique safe
center, every layer-two point must be covered.  Each of the six fibers
therefore requires at least two clauses, for a total of at least twelve,
contrary to \(\card S=10\).
\end{proof}

We finally treat the case in which the ten radius-one balls cover the cube.

\begin{proposition}\label{prop:ten-cover}
Ten distinct radius-one balls covering \([3]^4\) cannot have centers that
support a primitive circuit.
\end{proposition}

\begin{proof}
Fix an identification of the three question labels on each page with
\(\F_3\).  This choice does not affect Hamming distances.  Let \(S\) be the
set of ten centers and consider the ternary \([4,2,3]_3\)
Hamming code
\[
  C=\{(a,b,a+b,a+2b):a,b\in\F_3\}.
\]
Any two zero coordinates in a codeword force \(a=b=0\); hence every nonzero
codeword has weight at least three.  Since \((1,0,1,1)\in C\), the displayed
code has minimum distance three.  Its nine cosets partition \(\F_3^4\) into
codes of size nine.

Every radius-one ball meets each coset exactly once.  Indeed, the zero error
and the eight nonzero errors supported on one coordinate represent distinct
cosets of \(C\), because the difference of two such errors has weight at most
two.  Since \(\F_3^4/C\) has nine elements, these errors form a complete set
of coset representatives; translation by a ball center proves the claim.

Ten balls contribute ten incidences to each nine-point coset.  If they cover
the cube, every point has positive multiplicity, so each coset contains
exactly one double point and eight single points.  Across the nine cosets
there are exactly nine double points and no point of higher multiplicity.  If
\(n_d\) denotes the number of unordered pairs of centers at distance \(d\),
two radius-one balls meet in three points at center distance one, in two
points at distance two, and not at all at larger distances.  Counting the
nine double points gives
\begin{equation}\label{eq:intersection-app}
  3n_1+2n_2=9.
\end{equation}

Suppose \(S\) supports a primitive circuit.  Every page is ternary and,
by Lemma~\ref{lem:no-singleton-block}, every question block has size at least
two.  Among partitions of ten into three parts of size at least two,
\((4,3,3)\) minimizes the number of within-block pairs, giving
\[
  \binom42+\binom32+\binom32=12.
\]
Summed over four pages, the number of coordinate agreements satisfies
\[
  3n_1+2n_2+n_3
  \geq48.
\]
Together with \(n_1+n_2+n_3+n_4=45\), this implies
\[
  2n_1+n_2-n_4
  \geq3.
\]

With \(e_0,e_1,e_2\) denoting the standard basis of \(\R^3\), define
\[
  \Phi(q)=
  \left(
  \tfrac1{\sqrt3},
  e_{q_1}-\tfrac13\one,
  e_{q_2}-\tfrac13\one,
  e_{q_3}-\tfrac13\one,
  e_{q_4}-\tfrac13\one
  \right).
\]
The Gram matrix of the ten feature vectors satisfies
\[
  \mathcal G_{ij}
  =
  \inner{\Phi(q_i)}{\Phi(q_j)}
  =
  3-\dist(q_i,q_j).
\]
Centered and raw question indicators have the same column span after the
constant column is included.  Proposition~\ref{prop:rank-cap} therefore gives
\(\rank \mathcal G=\rank A_S=9\), and \(\mathcal G\) and
\(A_S^{\mathsf T}\) have the same
kernel on \(\R^S\).

The off-diagonal nonzero graph of \(\mathcal G\) joins pairs at distances one,
two, and four.  It must be connected.  Otherwise \(\mathcal G\) is block
diagonal, and
restricting the full-support circuit vector to each connected component gives
a nonzero kernel vector on that component.  Two components would yield two
independent kernel vectors, contrary to the one-dimensional circuit kernel.
A connected graph on ten vertices has at least nine edges, whereas
Eq.~\eqref{eq:intersection-app} gives
\[
  n_1+n_2+n_4
  =
  9-(2n_1+n_2-n_4)
  \leq6.
\]
This contradiction proves the claim.
\end{proof}

\section{Proof of the Cayley matching classification}
\label{app:cayley-classification}

We prove Theorem~\ref{thm:cayley-cyclic} here.

For \(n=1\), the matching classification is immediate.  The two indexed
clauses then have the same question tuple and opposite targets, so that
presentation is not reduced.  We henceforth assume \(n\geq2\) when referring
to the game.

We begin with the elementary word criterion underlying the matching
formulation.

\begin{lemma}\label{lem:cayley-chord}
Let \(W\) be a word of length \(2n\) in which each of \(n\) involutive
letters occurs exactly twice.  The word \(W\) freely reduces to the identity
if and only if the matching that joins equal letters is noncrossing in the
linear order of the word.
\end{lemma}

\begin{proof}
For a noncrossing matching, an innermost chord has consecutive endpoints.
Deleting that equal adjacent pair preserves noncrossing, and induction
reduces the word to the identity.  Conversely, suppose that \(W\) freely
reduces.  If \(W\) is nonempty, the first deletion in a free reduction
removes an adjacent equal pair.  Its chord crosses no other chord.  Delete
the pair; the shorter word freely reduces and its matching is noncrossing by
induction.  Reinserting the consecutive chord preserves noncrossing.
\end{proof}

In a length-\(2n\) Cayley word using every clause once, player \(\alpha\)
pairs \(E_x\) with \(O_{x\alpha^{-1}}\), because the question received on
\(O_g\) is \(g\alpha\).  As inversion permutes \(H\),
Lemma~\ref{lem:cayley-chord} identifies a true refutation of this form with a
common noncrossing order for the matching family \(\{m_h:h\in H\}\).

We first prove that a common noncrossing order forces \(H\) to be cyclic.
Suppose that such an order exists.  If its first
point is an \(O\)-point, interchange the two copies \(E\) and \(O\).
This sends \(m_h\) to \(m_{h^{-1}}\) and preserves the matching family.
A common left translation of both copies then makes the first point \(E_1\).
Thus no generality is lost by assuming that the order begins with \(E_1\).

\begin{lemma}\label{lem:cayley-side}
For every \(x\ne1\) and \(g\in H\),
\begin{equation}\label{eq:cayley-side}
  E_x<O_g
  \quad\Longleftrightarrow\quad
  O_{xg}<O_g.
\end{equation}
\end{lemma}

\begin{proof}
In \(m_g\), the chord incident with \(E_1\) is
\((E_1,O_g)\), while the chord incident with \(E_x\) is
\((E_x,O_{xg})\).  Since \(E_1\) is the first point and the two chords do
not cross, the endpoints \(E_x\) and \(O_{xg}\) lie on the same side of
\(O_g\).  This is Eq.~\eqref{eq:cayley-side}.
\end{proof}

List the \(O\)-points in their order of appearance:
\[
  O_{o_1}<O_{o_2}<\cdots<O_{o_n},
\]
and define \(S_j=\{x\in H:E_x<O_{o_j}\}\).  Applying
Eq.~\eqref{eq:cayley-side} with \(g=o_j\) gives
\begin{equation}\label{eq:Sj}
  S_j=\{1\}\cup\{o_i o_j^{-1}:1\leq i<j\}.
\end{equation}
The displayed elements are distinct, and none of the latter elements equals
the identity.  Hence \(\lvert S_j\rvert=j\) and
\[
  S_1\subsetneq S_2\subsetneq\cdots\subsetneq S_n.
\]

For \(1\leq j<n\), set \(d_j=o_j o_{j+1}^{-1}\).
Eq.~\eqref{eq:Sj} yields
\begin{equation}\label{eq:Sj-recurrence}
  S_{j+1}
  =
  \{1\}\cup(S_j\setminus\{1\})d_j\cup\{d_j\}.
\end{equation}

\begin{lemma}\label{lem:cayley-generator}
There is an element \(d\in H\) such that
\[
  d_j=d\quad(1\leq j<n),
  \qquad
  S_j=\{1,d,\ldots,d^{j-1}\}\quad(1\leq j\leq n).
\]
\end{lemma}

\begin{proof}
Set \(d=d_1\), so \(S_2=\{1,d\}\).  Suppose inductively that
\(S_j=\{1,d,\ldots,d^{j-1}\}\).  The recurrence gives
\[
  S_{j+1}
  =
  \{1,d_j,dd_j,\ldots,d^{j-1}d_j\}.
\]
Because \(d\in S_j\subset S_{j+1}\), either \(d=d_j\), or
\(d=d^k d_j\) for some \(1\leq k\leq j-1\).  The case \(k=1\) would give
\(d_j=1\), contrary to \(o_j\ne o_{j+1}\).  If \(k\geq2\), then
\[
  S_{j+1}
  =
  \{1\}\cup\{d^{1-k},d^{2-k},\ldots,d^{j-k}\}.
\]
The exponent interval contains zero, so the identity occurs both outside and
inside the second set; consequently \(\lvert S_{j+1}\rvert\leq j\), contrary
to \(\lvert S_{j+1}\rvert=j+1\).  Therefore \(d_j=d\), and
Eq.~\eqref{eq:Sj-recurrence} gives
\(S_{j+1}=\{1,d,\ldots,d^j\}\).
\end{proof}

Since \(S_n=H\) has \(n\) elements, Lemma~\ref{lem:cayley-generator}
implies that \(d\) has order \(n\) and \(H=\langle d\rangle\).  Moreover,
the new \(E\)-point entering between \(O_{o_j}\) and \(O_{o_{j+1}}\) is
\(E_{d^j}\), while \(o_{j+1}=o_jd^{-1}\).  Thus the common order is forced
to alternate between the two copies.

It remains to prove that cyclicity is sufficient.  Let
\(H=\langle d\rangle\) have order \(n\) and consider the
zero-indexed order
\begin{equation}\label{eq:zigzag}
  E_1,O_1,E_d,O_{d^{-1}},\ldots,
  E_{d^{n-1}},O_{d^{-(n-1)}}.
\end{equation}
Fix \(h=d^k\).  In the matching \(m_h\), the point \(E_{d^i}\) at position
\(2i\) is paired with the \(O\)-point at position \(2j+1\), where \(j\) is
the representative of \(-k-i\pmod n\) in \(\{0,\ldots,n-1\}\).  To compare
two chords, take \(i_1<i_2\), let \(j_s\) be the corresponding representative
for \(i_s\), and put \(\Delta=i_2-i_1\).

If \(j_2=j_1-\Delta\), then \(j_2<j_1\).  A crossing would require either
\[
  2i_1<2j_2+1<2j_1+1<2i_2
\]
or
\[
  2j_2+1<2i_1<2i_2<2j_1+1.
\]
The first set of inequalities gives \(j_2\geq i_1\) and \(j_1<i_2\);
substitution of \(j_2=j_1-\Delta\) instead gives \(j_1\geq i_2\).
The second gives \(j_2<i_1\) and \(j_1\geq i_2\); substitution instead
gives \(j_1<i_2\).  Both alternatives are impossible.

If the index wraps, then \(j_2=j_1+n-\Delta>j_1\).  Because
\(j_2<n\), one has \(\Delta>j_1\), so \(2i_2>2j_1+1\).
Moreover \(\Delta\leq n-1-i_1\) gives \(j_2\geq i_1+1\), and hence
\(2j_2+1>2i_1\).  The two alternating endpoint orders that could produce a
crossing are thereby excluded.  Thus \(m_h\) is noncrossing for every
\(h\in H\).

By Lemma~\ref{lem:cayley-chord}, the order
Eq.~\eqref{eq:zigzag} gives a true refutation.  The target parity is odd
because \(O_1\) occurs exactly once.  This proves
Theorem~\ref{thm:cayley-cyclic}.

\section{Lean correspondence}\label{app:lean}

\begin{figure}[h]
\centering
\resizebox{0.94\textwidth}{!}{%
\begin{tikzpicture}[
  result/.style={draw=blue!65!black, fill=blue!7, rounded corners,
    align=center, minimum height=8mm, text width=39mm},
  ingredient/.style={draw=black!60, fill=black!3, rounded corners,
    align=center, minimum height=7mm, inner xsep=3pt, font=\small},
  edge/.style={-{Stealth[length=2mm]}, semithick},
  companion/.style={edge, dashed}
]
\node[result] at (0,0) (main)
  {Theorem~\ref{thm:main}\\sharp threshold};
\node[result] at (-4.8,-1.5) (positive)
  {Theorem~\ref{thm:main}(i)\\three-question equivalence};
\node[result] at (7.8,-1.5) (negative)
  {Theorem~\ref{thm:main}(ii)\\Klein separation};

\node[ingredient] at (-8.4,-3.0) (weight)
  {Prop.~\ref{prop:weight-reduction}};
\node[ingredient] at (-4.8,-3.0) (obstruction)
  {Cor.~\ref{cor:obstruction-exactness}};
\node[ingredient] at (-7.2,-4.3) (alternating)
  {Lem.~\ref{lem:alternating-lift}};
\node[ingredient] at (-4.8,-4.3) (generation)
  {Prop.~\ref{prop:circuit-generation}};
\node[ingredient] at (-2.4,-4.3) (lifting)
  {Thm.~\ref{thm:exact-circuit-lifting}};

\node[ingredient] at (-10.0,-6.1) (nosingle)
  {Lem.~\ref{lem:no-singleton-block}};
\node[ingredient] at (-7.7,-6.1) (rank)
  {Prop.~\ref{prop:rank-cap}};
\node[ingredient] at (-5.4,-6.1) (multiplicity)
  {Lem.~\ref{lem:exact-multiplicity}};
\node[ingredient] at (-3.1,-6.1) (selected)
  {Prop.~\ref{prop:selected-block}};
\node[ingredient] at (-0.8,-6.1) (ternary)
  {Prop.~\ref{prop:four-ternary-reduction}};
\node[ingredient] at (1.5,-6.1) (subnine)
  {Prop.~\ref{prop:subnine}};
\node[ingredient] at (3.8,-6.1) (nine)
  {Prop.~\ref{prop:nine}};
\node[ingredient] at (6.1,-6.1) (tennoncover)
  {Prop.~\ref{prop:ten-noncover}};
\node[ingredient] at (8.4,-6.1) (tencover)
  {Prop.~\ref{prop:ten-cover}};

\node[ingredient] at (5.7,-3.0) (even)
  {Lem.~\ref{lem:even-subgroup}};
\node[ingredient] at (8.0,-3.0) (magnus)
  {Lem.~\ref{lem:magnus-conditions}};
\node[ingredient] at (10.8,-3.0) (cayley)
  {Thm.~\ref{thm:cayley-cyclic}};
\node[ingredient] at (8.4,-4.3) (chord)
  {Lem.~\ref{lem:cayley-chord}};
\node[ingredient] at (10.8,-4.3) (side)
  {Lem.~\ref{lem:cayley-side}};
\node[ingredient] at (13.2,-4.3) (generator)
  {Lem.~\ref{lem:cayley-generator}};

\draw[edge] (main) -- (positive);
\draw[edge] (main) -- (negative);
\draw[edge] (positive) -- (weight);
\draw[edge] (positive) -- (obstruction);
\draw[edge] (obstruction) -- (alternating);
\draw[edge] (obstruction) -- (generation);
\draw[edge] (obstruction) -- (lifting);
\draw[semithick] (lifting.south) -- (-2.4,-5.35);
\draw[semithick] (-10.0,-5.35) -- (8.4,-5.35);
\draw[edge] (-10.0,-5.35) -- (nosingle.north);
\draw[edge] (-7.7,-5.35) -- (rank.north);
\draw[edge] (-5.4,-5.35) -- (multiplicity.north);
\draw[edge] (-3.1,-5.35) -- (selected.north);
\draw[edge] (-0.8,-5.35) -- (ternary.north);
\draw[edge] (1.5,-5.35) -- (subnine.north);
\draw[edge] (3.8,-5.35) -- (nine.north);
\draw[edge] (6.1,-5.35) -- (tennoncover.north);
\draw[edge] (8.4,-5.35) -- (tencover.north);
\draw[edge] (negative) -- (even);
\draw[edge] (negative) -- (magnus);
\draw[companion] (negative) -- (cayley);
\draw[edge] (cayley) -- (chord);
\draw[edge] (cayley) -- (side);
\draw[edge] (cayley) -- (generator);
\end{tikzpicture}%
}
\caption{Core dependency graph for the Lean formalization.  Solid arrows
point from a conclusion to its principal ingredients.  The dashed arrow
marks the companion each-clause-once classification.}
\label{fig:lean-dependencies}
\end{figure}
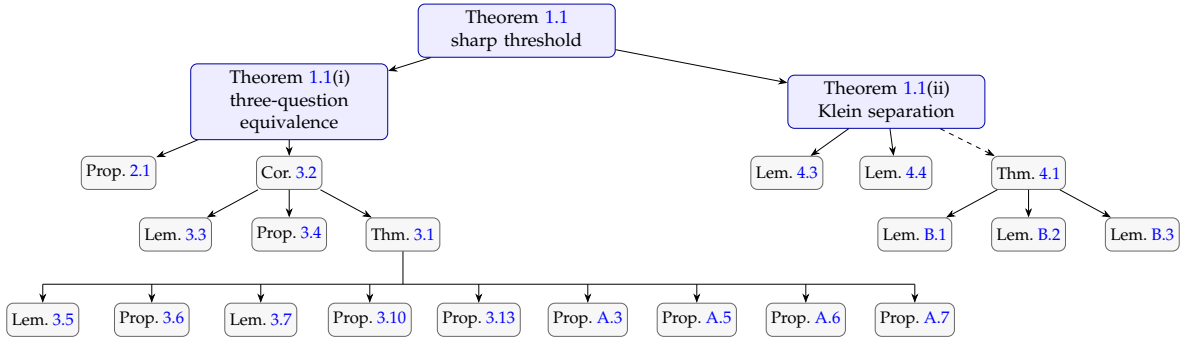

The accompanying development~\cite{FourPlayerXORGames2026} formalizes the
principal threshold theorem and its obstruction-space formulation in
Lean~4~\cite{deMouraUllrich2021}, using Mathlib~\cite{mathlib2020} and the
Lean-QIT library~\cite{ZhuEtAl2026LeanQIT}.  Lean-QIT is imported as the
external Lake package \texttt{QIT}; it is not vendored into the accompanying
repository.  Figure~\ref{fig:lean-dependencies} displays the dependency graph
for the two parts of Theorem~\ref{thm:main}.  Table~\ref{tab:lean-correspondence}
gives the statement-by-statement correspondence.

The status \emph{direct} means that the Lean declaration states the
manuscript conclusion up to notation and packaging.  \emph{Partial} marks a
proper part of a manuscript statement, \emph{supporting} marks a conditional
component or a consequence used by a directly formalized result, and
\emph{alternative} marks an application-specific combinatorial route in place
of the more general lemma stated in the manuscript.  This distinction is
material for the exact-multiplicity circuit statement, the refutation clause
of Theorem~\ref{thm:cayley-cyclic}, and the general Magnus lemma: their
downstream conclusions are kernel checked, but the labeled intermediate
statements are not each reproduced verbatim in Lean.

\begingroup
\footnotesize
\setlength{\tabcolsep}{0pt}
\renewcommand{\arraystretch}{1.08}
\newcommand{\leanbreak}{\par\kern1.5pt}
\begin{longtable}{@{}
  >{\raggedright\arraybackslash}p{0.20\textwidth}
  @{\hspace{1.4em}}
  >{\raggedright\arraybackslash}p{0.58\textwidth}
  @{\hspace{1.6em}}
  >{\raggedright\arraybackslash}p{0.105\textwidth}
  @{}}
\caption{Correspondence between manuscript statements and Lean
declarations.}\label{tab:lean-correspondence}\\
\toprule
Manuscript statement & Principal Lean declaration(s) & Status \\
\midrule
\endfirsthead
\multicolumn{3}{c}{\tablename~\thetable\ (continued)}\\
\toprule
Manuscript statement & Principal Lean declaration(s) & Status \\
\midrule
\endhead
\midrule
\multicolumn{3}{r}{Continued on next page}\\
\endfoot
\bottomrule
\endlastfoot
Theorem~\ref{thm:main}(i) &
\path{QIT.XORGame.commutingOperatorValue_eq_one_iff_hasPerfectMERP_of_activeQuestionBound_le_three}
& Direct \\
Theorem~\ref{thm:main}(ii) and sharpness &
\path{QIT.XORGame.v4_value_one}\leanbreak
\path{QIT.XORGame.v4_no_MERP}\leanbreak
\path{QIT.XORGame.four_is_sharp_threshold}
& Direct \\
Proposition~\ref{prop:weight-reduction} &
\path{QIT.XORGame.commutingOperatorValue_eq_one_iff_of_support_eq}\leanbreak
\path{QIT.XORGame.weightedValue_lt_one_of_uniformValue_lt_one}
& Partial \\
Theorem~\ref{thm:exact-circuit-lifting} &
\path{QIT.XORGame.primitiveCircuit_lifts_fourByThree}
& Supporting \\
Corollary~\ref{cor:obstruction-exactness} &
\path{QIT.XORGame.obstructionSpaces_eq_fourByThree}
& Direct \\
Lemma~\ref{lem:alternating-lift} &
\path{QIT.XORGame.balancedWord_alternatingLift_isIntegerRelation}\leanbreak
\path{QIT.XORGame.alternatingLift_mod_two}
& Direct \\
Proposition~\ref{prop:circuit-generation} &
\path{QIT.XORGame.primitiveCircuits_closure_eq_kernel}\leanbreak
\path{QIT.XORGame.everyIntegerRelationHasBalancedLift_of_primitiveCircuits}
& Direct \\
Lemma~\ref{lem:no-singleton-block} &
\path{QIT.XORGame.fullCircuit_questionFiber_card_ne_one}
& Direct \\
Proposition~\ref{prop:rank-cap} &
\path{QIT.XORGame.fullRationalCircuit_card_eq_rank_add_one}\leanbreak
\path{QIT.XORGame.rationalIncidence_rank_le_usedQuestions}\leanbreak
\path{QIT.XORGame.fourByThree_primitiveCircuit_card_le_ten}
& Direct \\
Lemma~\ref{lem:exact-multiplicity} &
\path{QIT.XORGame.parityPairing_occurrenceParity}\leanbreak
\path{QIT.XORGame.exactMultiplicityBalancedLift_hasBalancedLift}
& Supporting \\
Lemma~\ref{lem:ordering} &
\path{QIT.XORGame.InvolutionWord.reducesToEmpty_of_toFinset_card_le_two_of_alternatingSum_zero}\leanbreak
\path{QIT.XORGame.exists_bipartiteLinearForestAlternatingLayout}
& Direct \\
Proposition~\ref{prop:selected-block} &
\path{QIT.XORGame.selectedBlockLinearForest_exactMultiplicityBalancedLift}
& Direct \\
Proposition~\ref{prop:three-color} &
\path{QIT.XORGame.threeColor_colorwiseMatching_linearForest}
& Direct \\
Proposition~\ref{prop:four-singleton} &
\path{QIT.XORGame.ColorSignatureSystem.linearForestMatching_of_singleton}
& Direct \\
Proposition~\ref{prop:four-ternary-reduction} &
\path{QIT.XORGame.hasExactMultiplicityBalancedLift_of_ternaryPlayers_card_le_three}
& Direct \\
Lemma~\ref{lem:separation-packing} &
\path{QIT.XORGame.safeCenters_card_le_six_of_even_separated}\leanbreak
\path{QIT.XORGame.evenSupportDistances_separated_of_sharedBlockMissingCell_escape}
& Partial \\
Lemma~\ref{lem:safe-excess} &
\path{QIT.XORGame.safeHole_neighborhood_excess}
& Direct \\
Proposition~\ref{prop:subnine} &
\path{QIT.XORGame.atMostEight_exactLift_of_nonliftable_even_separated}
& Supporting \\
Lemma~\ref{lem:nine-independence} &
\path{QIT.XORGame.nineBallCover_rationalIncidence_linearIndependent}
& Direct \\
Proposition~\ref{prop:nine} &
\path{QIT.XORGame.nineClause_exactLift_of_nonliftable_even_separated}
& Supporting \\
Proposition~\ref{prop:ten-noncover} &
\path{QIT.XORGame.tenClause_noncovering_geometry_impossible}
& Supporting \\
Proposition~\ref{prop:ten-cover} &
\path{QIT.XORGame.tenBallCover_structural_contradiction}\leanbreak
\path{QIT.XORGame.tenBallCover_not_fullRationalCircuit}
& Supporting \\
Theorem~\ref{thm:cayley-cyclic} &
\path{QIT.Research.FourXOR.hasCommonNoncrossingOrder_iff_isCyclic}\leanbreak
\path{QIT.Research.FourXOR.cyclic_cayley_has_true_refutation}
& Partial \\
Lemma~\ref{lem:cayley-chord} &
\path{QIT.Research.FourXOR.NestedWord.equalPairs_noncrossing}\leanbreak
\path{QIT.Research.FourXOR.NestedWord.reducible}
& Supporting \\
Lemma~\ref{lem:cayley-side} &
\path{QIT.Research.FourXOR.same_side}
& Direct \\
Lemma~\ref{lem:cayley-generator} &
\path{QIT.Research.FourXOR.S_succ_dj_eq_d}\leanbreak
\path{QIT.Research.FourXOR.S_pow}
& Partial \\
Lemma~\ref{lem:even-subgroup} &
reduction interfaces in\leanbreak
\path{FourXOR.AllLength}\leanbreak
\path{FourXOR.AllLengthGame}
& Alternative \\
Lemma~\ref{lem:magnus-conditions} &
\path{QIT.Research.FourXOR.altSum_reducible}\leanbreak
\path{QIT.Research.FourXOR.pairSum_reducible}
& Alternative \\
Elementary-abelian all-length conclusion &
\path{QIT.Research.FourXOR.cayley_f2r_hasPREF_and_no_refutation}
& Direct \\
\end{longtable}
\endgroup

The source tree contains no \texttt{sorry}, \texttt{admit}, or manually
declared axiom.  The repository additionally runs \texttt{\#print axioms}
on the paper-facing declarations.  Their dependency sets contain Lean's
standard \texttt{propext}, \texttt{Classical.choice}, and
\texttt{Quot.sound}.  The three-question circuit argument also uses seven
explicitly enumerated \texttt{native\_decide} bridge axioms for finite
cardinality and Hamming-space computations; the Klein witness, translation
matching classification, and elementary-abelian all-length theorem use only
the three standard axioms.  The tracked toolchain, Lake manifest, exact
certificate scripts, and full declaration map are included in the
accompanying repository~\cite{FourPlayerXORGames2026}.

\end{document}